\documentclass[11pt]{article}
\usepackage[margin=1in]{geometry}
\usepackage{amsmath, amssymb, amsthm, amsfonts,cite,alltt,clrscode}
\usepackage{array,colortbl,hhline,xcolor}
\usepackage[english]{babel}
\usepackage{graphicx}
\usepackage{caption}
\usepackage{url}
\usepackage{float}
\usepackage{wrapfig}
\usepackage[list=true,listformat=simple]{subcaption}
\usepackage{enumerate}

\usepackage{hyperref}

\def\emph#1{\textbf{\textit{\boldmath #1}}}

\def\NOR{\mathbin{\textsc{nor}}}

\newtheorem{theorem}{Theorem}[section]
\newtheorem{lemma}[theorem]{Lemma}

\newtheorem{corollary}[theorem]{Corollary}
\newtheorem{definition}[theorem]{Definition}

\graphicspath{{./figures/simulate_l2t/}{./figures/planarity/}{./figures/NCL/}{./figures/locking_2-toggle/}{./figures/2-toggle-paper/}{./figures/}}

{\makeatletter
 \gdef\xxxmark{%
   \expandafter\ifx\csname @mpargs\endcsname\relax 
     \expandafter\ifx\csname @captype\endcsname\relax 
       \marginpar{xxx}
     \else
       xxx 
     \fi
   \else
     xxx 
   \fi}
 \gdef\xxx{\@ifnextchar[\xxx@lab\xxx@nolab}
 \long\gdef\xxx@lab[#1]#2{\textbf{[\xxxmark #2 ---{\sc #1}]}}
 \long\gdef\xxx@nolab#1{\textbf{[\xxxmark #1]}}
}

\let\realbibitem=\bibitem
\def\bibitem{\par \vspace{-1.2ex}\realbibitem}

\begin{document}

\title{Trains, Games, and Complexity: \\
  0/1/2-Player Motion Planning through Input/Output Gadgets}

\author{%
  Hayashi Ani%
  \thanks{MIT Computer Science and Artificial Intelligence Laboratory, 32 Vassar Street, Cambridge, MA 02139, USA, \protect\url{{edemaine,dylanhen,jaysonl}@mit.edu}}
\and
  Erik D. Demaine%
    \footnotemark[1]
\and
  Dylan H. Hendrickson\footnotemark[1]
\and
  Jayson Lynch\footnotemark[1]
}

\date{}
\maketitle

\begin{abstract}
We analyze the computational complexity of motion planning
through local ``input/output'' gadgets with separate entrances and exits,
and a subset of allowed traversals from entrances to exits, each of which
changes the state of the gadget and thereby the allowed traversals.
We study such gadgets in the zero-, one-, and two-player settings, in particular
extending past motion-planning-through-gadgets work \cite{Toggles_FUN2018,DHL}
to zero-player games for the first time, by considering ``branchless'' connections
between gadgets that route every gadget's exit to a unique gadget's entrance.
Our complexity results include containment in L, NL, P, NP, and PSPACE; as well as hardness for NL, P, NP, and PSPACE.
We apply these results to show PSPACE-completeness for certain mechanics in the video games Factorio, [the Sequence], and a restricted version of Trainyard, improving the result of \cite{trainyard2}.
This work strengthens prior results on switching graphs, ARRIVAL~\cite{Arrival}, and reachability switching games \cite{switchinggames}.
\end{abstract}

\section{Introduction}
\label{sec:intro}

Imagine a train proceeding along a track within a railroad network.
Tracks are connected together by ``switches'':
upon reaching one, the switch chooses the train's next track deterministically
based on the state of the switch and where the train entered the switch;
furthermore, the traversal changes the switch's state,
affecting the next traversal.
ARRIVAL~\cite{Arrival} is one game of this type, where every switch has a
single input and two outputs, and alternates between sending the train along
the two outputs; the goal is to determine whether the train ever reaches
a specified destination.  Even this seemingly simple game has unknown
complexity, but is known to be in NP $\cap$ coNP \cite{Arrival},
so cannot be NP-hard unless NP${}={}$coNP. More recent work shows a stronger result of containment in UP $\cap$ coUP as well as CLS \cite{arrivalcls}, PLS \cite{karthik2017did}, and UEOPL \cite{fearnley2020unique}.
But what about other types of switches?

In this paper, we introduce a very general notion of ``input/output gadgets''
that models the possible behaviors of a switch, and analyze the resulting
complexity of motion planning/prediction
(does the train reach a desired destination?)\ while navigating a network of switches/gadgets.
This framework gives us an expressive set of problems for
various complexity classes to use as starting points for
hardness reductions to other problems of interest.
For example, it is related to the ``reachability switching games'' of  \cite{switchinggames}, which in turn generalize ``switching systems'' known as Propp machines.
In addition to ARRIVAL, our framework captures other toy-train models,
including those in the video games Factorio and Trainyard.
In many cases, we obtain PSPACE-hardness, enabling building of a
(polynomial-space) computer out of a deterministic railway system
with a single train.
Intuitively, our model is similar to a circuit model of computation, but where
the state is stored in the gates (gadgets) instead of the wires, and gates
update only according to visits by a single deterministically controlled
agent (the train).

This work builds off of prior work on the computational complexity of agent-based motion planning \cite{Toggles_FUN2018, DHL}, extending it to zero-player situations. An analogous generalization of computational problems based on the number of players and boundedness of moves can be found in Constraint Logic \cite{hearn2009games}, which has served as a framework for a large number of hardness proofs for reconfiguration problems as well as games and puzzles.

\subsection{Motion Planning through Gadgets}

Our model is a natural zero-player adaptation of the
\emph{motion-planning-through-gadgets} framework developed in \cite{DHL}
(after its introduction at FUN 2018 \cite{Toggles_FUN2018}),
so we begin with a summary of that framework.
A \emph{gadget} $G = (Q,L,T)$ consists of
a finite set $Q$ of \emph{states},
a finite set $L$ of \emph{locations} (entrances/exits),
and a set $T \subseteq (Q \times L)^2$ of \emph{transitions}
of the form $(q, a) \to (r, b)$ where $q,r \in Q$ and $a,b \in L$.
Figure~\ref{fig:w/su/sd state diagram} shows an example of a gadget.
A transition $(q, a) \to (r, b) \in T$ means that,
when the gadget is in state $q$, an agent can \emph{traverse}
the gadget by entering the gadget at location $a$ and exiting at location~$b$,
while changing the state of the gadget from $q$ to~$r$.
In general, a location might serve as the entrance for one traversal
and the exit for another traversal.  In this paper, however, we consider
the special case (as in Figure~\ref{fig:w/su/sd state diagram})
where each location serves exclusively as an entrance or an exit for the agent,
but not both;
our figures will usually put entrances (which we call inputs) on the left,
and put exits (outputs) on the right.

\begin{figure}
  \centering
  \includegraphics[scale=.8]{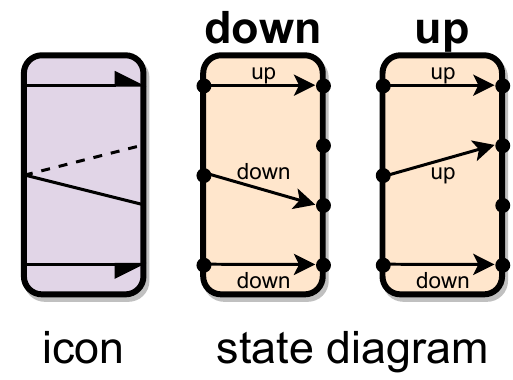}
  \caption{An example gadget---the switch/set-up line/set-down line
    of Figure~\ref{fig:basis w/su/sd}---which is a 7-location 2-state
    input/output gadget.  The agent can enter at any of the three input
    locations on the left, and exit at the corresponding output location
    on the right.  Traversing the top or bottom line sets the gadget's state
    to `up' or `down', respectively, which controls the output of the
    middle traversal to be the top or bottom, respectively,
    of its two options.  The middle traversal does not change the state.}
  \label{fig:w/su/sd state diagram}
\end{figure}

We can think of a gadget as a graph $(Q \times L, T)$
on state/location pairs, called the \emph{transition graph}.
We sometimes also consider the \emph{state-transition graph} of a gadget,
which is the directed multigraph with a vertex for each state $\in Q$ and a
directed edge $(q,r)$ for each transition $(q,a) \to (r,b) \in T$
for any $a,b \in L$.
In figures such as Figure~\ref{fig:w/su/sd state diagram},
we define gadgets using a \emph{state diagram} which gives,
for each state $q \in Q$, a labeled directed multigraph $G_q = (L, E_q)$
on the locations, where a directed edge $(a,b)$ with label $r$ represents
a transition $(q,a) \to (r,b) \in T$
(and thus $G_q$ represents the available transitions in state~$q$).

A \emph{system of gadgets} consists of a set of gadgets, their initial states,
and a \emph{connection graph} on the gadgets' locations.
If two locations $a,b$ of two gadgets (possibly the same gadget) are connected
by a path in the connection graph, then an agent can traverse freely between
$a$ and~$b$ (outside the gadgets).
(Equivalently, we can think of locations $a$ and $b$ as being identified,
effectively contracting connected components of the connection graph.)
Gadgets are \emph{local} in the sense that traversing a gadget does
not change the state of any other gadgets.

In \emph{one-player motion planning}, we are given initial and goal locations
$s,t$ of a single agent in a system of gadgets,
and the problem asks whether there is a
sequence of traversals that brings the agent from $s$ to~$t$.
Two-player and team motion planning are also introduced in \cite{DHL},
but not discussed here.

Past work \cite{DHL} analyzed (and in many cases, characterized)
the complexity of these motion-planning problems
for gadgets satisfying a few additional properties, specifically,
gadgets that are ``reversible deterministic $k$-tunnel'' or that are
``DAG $k$-tunnel'', defined as follows:

\begin{itemize}
\item A gadget is \emph{$k$-tunnel} if there is a perfect matching on
  its $2k$ locations, whose matching edges are called \emph{tunnels},
  such that the gadget only allows traversals between endpoints of a tunnel.

\item A gadget is \emph{deterministic}
  if its transition graph has maximum out-degree $\leq 1$,
  i.e., an agent entering the gadget in some state $q$
  at some location $a$ can exit in only one state $r$ and
  at only one location~$b$.

\item A gadget is \emph{reversible} if its transition graph has the reverse of
  every edge, i.e., every traversal could be immediately undone.

\item A gadget is a \emph{DAG} if it has an acyclic \textit{state-transition}
  graph, i.e., no sequence of traversals repeats a state.
  Such gadgets can necessarily be traversed only a bounded number of times
  (at most the number of states).
\end{itemize}

\subsection{Input/Output Gadgets and Zero-Player Motion Planning}


We define a gadget to be \emph{input/output} if its locations can be
partitioned into \emph{input} locations (entrances) and
\emph{output} locations (exits) such that every traversal brings an agent
from an input location to an output location, and in every state,
there is at least one traversal from each input location.
In particular, deterministic input/output gadgets have exactly one traversal
from each input location in each state.
Note that input/output gadgets cannot be reversible nor DAGs,
so prior characterizations \cite{DHL} do not apply to this setting.
Indeed, the example of Figure~\ref{fig:w/su/sd state diagram} satisfies none
of the $k$-tunnel, deterministic, reversible, or DAG properties.

An input/output gadget is \emph{output-disjoint} if, for each output location,
all of the transitions to it (including those from different states) are from the same input location.
This condition is a generalization of $k$-tunnel: it allows a one-to-many
relation from a single input to multiple outputs.


With deterministic input/output gadgets, we can define a natural
\emph{zero-player motion-plan\-ning game} as follows.
A system of input/output gadgets is \emph{branchless}
if each connected component of the
connection graph contains at most one input location.
Intuitively, if an agent finds itself in such a connected component,
then there is only one gadget location it can enter, uniquely defining how
it should proceed.
(If an agent finds itself in a connected component with no input locations,
it is stuck in a dead-end and the game ends.)
We can think of edges in the connection graph as directed wires that point
from output locations to the input location in the same connected component.
Note that branchless systems can still have multiple output locations in a
connected component, which functions as a fan-in.%
\footnote{In fact, the original framework of \cite{Toggles_FUN2018} was
  inherently branchless: connections between locations formed a matching.
  That framework used a 1-state nondeterministic ``branching hallway'' gadget
  to connect multiple locations to each other.
  For branchless input/output systems, we can equivalently think of
  replacing the branching hallway with a 1-state ``fan-in'' input/output gadget
  with traversals from two inputs to one output.}

In a branchless system of deterministic input/output gadgets, there are never
any choices to make: in the connection graph, there is at most one reachable
input location, and when the agent enters a gadget at an input location,
there is exactly one transition it can make.
Thus we define \emph{zero-player motion planning}
with a set of deterministic input/output gadgets to be the one-player
motion-planning problem restricted to branchless systems of those gadgets.
Lacking any agency, the decision problem is equivalent to whether the agent
ever reaches the goal location while following the unique path available to it
(before cycling or hitting a dead-end).

\subsection{Classifying Output-Disjoint Deterministic 2-State Input/Output Gadgets}\label{sec:intro characterization}

In this paper, we are primarily interested in output-disjoint deterministic 2-state input/output gadgets.
In this section, we omit the adjectives and refer to them simply as ``gadgets'';
and categorize these gadgets as ``trivial'', ``bounded'',
or ``unbounded''.

The behavior of an input location to a gadget is described by how it changes the state and which output location it sends the agent to in each state.
If the input location does not change the state and always uses the same output location, it can be ignored (the agent's path can be ``shortcut'' to skip that transition); we call this a \emph{trivial line}.
Otherwise, the input location corresponds to one of the five nontrivial subunits shown in Table~\ref{tab:subunits}.
A gadget is then a disjoint union of some of these subunits;
Figures~\ref{fig:bounded basis} and \ref{fig:unbounded basis} show some different ways these subunits can be assembled into different gadgets.

\begin{table}
\centering
\begin{tabular}{llm{0.6\textwidth}}
  $\vcenter{\hbox{\includegraphics[scale=.8]{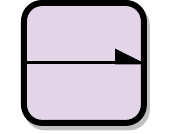}}}$ &
  \bf Set-Up Line &
  A tunnel that can always be traversed in one direction and sets the state of the gadget to a specific state (`up').
  \\[1ex]
  $\vcenter{\hbox{\includegraphics[scale=.8]{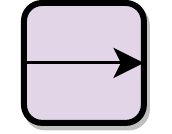}}}$ &
  \bf Toggle Line &
  A tunnel that can always be traversed in one direction and toggles the state with each crossing.
  \\[1ex]
  $\vcenter{\hbox{\includegraphics[scale=.8]{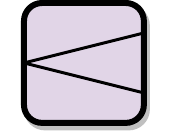}}}$ &
  \bf Switch &
  A three-location gadget with one input which transitions to one of two
  outputs (`top' or `bottom') depending on the state (`up' or `down'
  respectively), without changing the state.
  \\[1ex]
  $\vcenter{\hbox{\includegraphics[scale=.8]{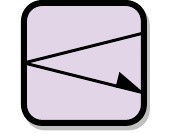}}}$ &
  \bf Set-Up Switch &
  A switch that also sets the state of the gadget to a specific state (`up').
  \\[1ex]
  $\vcenter{\hbox{\includegraphics[scale=.8]{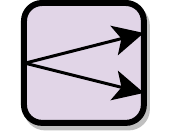}}}$ &
  \bf Toggle Switch &
  A switch that also toggles the state of the gadget with each crossing.
\end{tabular}
\caption{The five possible subunits (modulo up/down symmetry)
  for output-disjoint deterministic 2-state input/output gadgets,
  whose states are named `up' and `down'.
  In general, unannotated lines denote transitions that do not change the state,
  full arrowheads denote transitions that always toggle the state,
  and half arrowheads denote transitions that always set the state to a
  specific value (`up' or `down' according to the half arrowhead).}
\label{tab:subunits}
\end{table}

We call the states of a two-state gadget \emph{up} and \emph{down}, and
assume that each switch transitions to the top output in the up state
and the bottom output in the down state;
because we are not concerned with planarity,
this assumption is fully general by possible reflection of each subunit.
There are two versions of the set line and set switch:
one that sets the gadget to each state, up or down.
For example, a gadget with a set-up switch and set-up line
(Figure~\ref{fig:basis suw/su})
is meaningfully different from a gadget with a
set-up switch and set-down line (Figure~\ref{fig:basis suw/sd}).
We draw the set-down line and switch as the
reflections of the set-up version in Table~\ref{tab:subunits}.
To represent the current state of a gadget,
we draw one of the lines in each switch dashed,
so that the next transition would be made along a solid line.

The ARRIVAL problem \cite{Arrival} is equivalent to zero-player motion planning
with just the toggle switch from Table~\ref{tab:subunits}: each vertex in
their switch graph corresponds to a toggle switch in a system of gadgets.
We will use their terminology when referring to switch graphs in ARRIVAL
\cite{Arrival};
however, when referring to gadgets in our model, a switch is a gadget
(or part of a gadget) which does not change state when traversed
(as in Table~\ref{tab:subunits}).
More generally, zero-player motion planning with an arbitrary set of
deterministic single-input input/output gadgets (with gadgets specified as part
of the instance) is equivalent to explicit zero-player reachability switching
games, as defined in \cite{switchinggames}.

We categorize gadgets into three families:
\begin{enumerate}
\item \emph{Trivial} gadgets have either no state change or no state-dependent behavior; they are composed entirely of switches or entirely of toggle and set lines. Trivial gadgets are equivalent to (collections of) trivial lines, or equivalently always-open tunnels. Zero-player motion planning with trivial gadgets is in L by straightforwardly simulating the agent for a number of steps equal to the number of locations.

\item \emph{Bounded} gadgets have state-dependent behavior (i.e., some kind of switch) and have only one-way state change, either only to the up state or only to the down state. A bounded gadget can change its state at most once, so such gadgets naturally give rise to bounded games in which the maximum number of moves is polynomially bounded.

\item \emph{Unbounded} gadgets have state-dependent behavior (switches)
and have transitions that change state in both directions.
For example, the toggle switch of ARRIVAL is unbounded.
Unbounded gadgets naturally give rise to unbounded games
in which the number of moves can be exponential.
\end{enumerate}

We will find that the complexity of motion planning with a given gadget
also depends on whether the gadget is \emph{single-input}
or \emph{multi-input},
where we count only ``nontrivial'' input locations.
A \emph{nontrivial input} must have a transition from that input
that either changes the state of the gadget
or does not exist in all states of the gadget.
The only nontrivial single-input gadgets are the set switch and toggle switch,
which are bounded and unbounded, respectively.

\subsection{Our Results: Complexity}

Table~\ref{tbl:2-state} summarizes our main complexity results for
zero-player motion planning with output-disjoint deterministic 2-state
input/output gadgets.  While our main motivation was to analyze
zero-player motion planning, we also characterize the complexity of
one-player motion planning for contrast.
These complexity results apply to \textit{any} gadget in the family
specified in each column, and more generally to any nonempty set of gadgets
in the family
(optionally with gadgets from simpler families in leftward columns).
In particular, we prove that motion planning with \textit{any}
multi-input bounded gadget(s) (and optionally with trivial gadgets)
is P-complete for zero-player and NP-complete for one-player;
while motion planning with \textit{any} multi-input unbounded gadget(s)
(and optionally with trivial or bounded gadgets)
is PSPACE-complete for both zero- and one-player.

\begin{table}
  \centering
  \newcolumntype{P}[1]{>{\raggedright\parskip=8pt}p{#1}}
  \def\arraystretch{1.4}
  \def\class#1{#1}
  \def\vs{~\textsc{vs.}~}
  \definecolor{purple}{rgb}{0.47,0.25,0.55}
  \def\HEADER{\bfseries\cellcolor{purple!85}\textcolor{white}}
  \arrayrulecolor{purple!50!black}
  \setlength\arrayrulewidth{1pt}
  \def\REF#1{ [\S\ref{#1}]}
  \begin{tabular}{|P{10em}|P{7em}|P{9em}|P{9em}|}
    \hhline{~|-|-|-|}
      \multicolumn{1}{c|}{}
    & \HEADER{Trivial (always-open tunnels)}
    & \HEADER{Bounded \& multiple \\ nontrivial inputs}
    & \HEADER{Unbounded \& multiple \\ nontrivial inputs}
    \cr
    \hline
    \HEADER{Zero-player (fully deterministic)\REF{sec:0p}}
    & \class{L}
    & \class{P-complete}
    & \class{PSPACE-complete}
    \cr
    \hline
    \HEADER{One-player\REF{sec:1p}}
    & \class{NL-complete}
    & \class{NP-complete} 
    & \class{PSPACE-complete}
    \cr
    \hline
  \end{tabular}
  \caption{Complexity of zero- and one-player motion planning
    for arbitrary output-disjoint deterministic 2-state input/output gadget(s),
    with multiple nontrivial inputs in nontrivial gadgets.}
  \label{tbl:2-state}
\end{table}

Table~\ref{tbl:1 input} summarizes our results for motion planning
with \textit{single-input} nontrivial input/output gadgets.
This case is a more immediate generalization of ARRIVAL \cite{Arrival},
and is equivalent to the reachability switching games
studied in \cite{switchinggames}.
We strengthen the results of \cite{switchinggames} in two ways.
First, we show that the containments in NP and EXPTIME
still hold when we allow nondeterministic gadgets.
Second, we show hardness for specific constant-size gadgets---the toggle switch
for zero-player, and each of the toggle switch and set switch
for one- and two-player---instead of having unbounded-size gadgets
specified as part of the instance.
In particular, these hardness results apply to all (two)
nontrivial single-input gadgets for one- and two-player;
the complexity of the set switch for zero-player remains open.

\begin{table}
  \centering
  \newcolumntype{P}[1]{>{\raggedright\parskip=8pt}m{#1}}
  \def\arraystretch{1.4}
  \def\class#1{#1}
  \definecolor{purple}{rgb}{0.47,0.25,0.55}
  \def\HEADER{\bfseries\cellcolor{purple!85}\textcolor{white}}
  \arrayrulecolor{purple!50!black}
  \setlength\arrayrulewidth{1pt}
  \def\REF#1{ [\S\ref{#1}]}
  \def\CITE#1{ \cite{#1}}
  \begin{tabular}{|P{10em}|P{12em}|P{12em}|}
    \hhline{~|-|-|}
    \multicolumn{1}{c|}{}
    & \HEADER{Contained in}
    & \HEADER{Hard for}
    \cr
    \hline
    \HEADER{Zero-player (fully deterministic)\REF{sec:0p}}
    & \class{UP} $\cap$ \class{coUP} \CITE{arrivalcls}
    & \class{NL} for toggle switch \REF{sec:0p 1input} (cf.~\cite{switchinggames})
    \cr
    \hline
    \HEADER{One-player\REF{sec:1p}}
    & \class{NP} \REF{sec:1p in np} (cf.~\cite{switchinggames})
    & \class{NP}  \REF{sec:1p np-hard} (cf.~\cite{switchinggames})
    \cr
    \hline
    \HEADER{Two-player\REF{sec:2p}}
    & \class{EXPTIME} \REF{sec:2p} (cf.~\cite{switchinggames})
    & \class{PSPACE} \REF{sec:2p} (cf.~\cite{switchinggames})
    \cr
    \hline
  \end{tabular}
  \caption{Complexity results for zero-, one-, and two-player motion planning
    with any nontrivial single-input input/output gadget(s)
    (the toggle switch and/or the set switch).}
  \label{tbl:1 input}
\end{table}

Our complexity results
for zero-player, one-player, and two-player motion planning
are presented in Sections~\ref{sec:0p}, \ref{sec:1p}, and \ref{sec:2p},
respectively.

In Section~\ref{Applications}, we apply our input/output gadget framework
to prove PSPACE-completeness of mechanics in several video games:
one-train colorless Trainyard, the game [the Sequence],
trains in Factorio, and transport belts in Factorio
are all PSPACE-complete.
The first result improves a previous PSPACE-completeness result for two-color
Trainyard \cite{trainyard2} by using a strict subset of game features.
Factorio in general is trivially PSPACE-complete, as players have explicitly
built computers using the circuit network; here we prove hardness for
the restricted problems with only train-related objects and only
transport-belt-related objects.

\subsection{Our Results: Simulation}

How do we prove that zero-player motion planning with \textit{any}
multi-input bounded or unbounded gadget is P-complete or PSPACE-complete,
respectively?
We show how to reduce these infinite families down to finitely many cases
through the concept of ``simulation''.

A \emph{zero-player simulation} of a gadget $G$
is a branchless system of gadgets,
together with a mapping of input and output locations of $G$
to distinct input and output locations of gadgets in the system,
that has the same behavior as $G$ in the natural sense:
if the agent enters the system at a sequence of input locations
corresponding to inputs of~$G$,
then the system sends the agent to the output locations
corresponding to the outputs $G$ would send the agent to.
(Note that some locations of the system may not correspond
to any locations of~$G$.)
We say that $G'$ \emph{simulates} $G$
if there is a system of $G'$ gadgets that is a simulation of~$G$.

This definition of simulation is only applicable to zero-player motion planning,
and thus with deterministic input/output gadgets.
We can define a similar notion of simulation for one-player motion planning;
see also \cite{hendrickson2021gadgets,GadgetsChecked_FUN2022} for more
precise definitions.
A \emph{one-player simulation} of a gadget $G$ is a system of gadgets,
together with a mapping of locations of $G$ to distinct locations
of gadgets in the system, that has the same behavior as $G$
in the natural nondeterministic sense:
if the agent enters the system at a sequence of $k$ locations
corresponding to locations of~$G$,
then the agent can exit the system
in a sequence of $k$ locations corresponding to locations of $G$
if and only if it could have made the corresponding sequence of traversals in $G$.

Any zero-player simulation is also a one-player simulation,
so all of our results for zero-player simulations
immediately carry over to the one-player case.

Crucially, simulations yield logarithmic-space polynomial-time reductions:
simply replace each copy of $G$ with a copy of the system simulating it.
In particular, simulations preserve hardness of
zero-player and one-player motion planning
for NL, P, NP, and PSPACE.

To characterize all multi-input input/output gadgets,
we show that they all simulate at least one of the eight gadgets
listed in Lemma~\ref{lem:gadget basis}
and shown in Figures~\ref{fig:bounded basis} (bounded)
and~\ref{fig:unbounded basis} (unbounded),
and thus it will suffice to show hardness for these eight cases.

\begin{figure}
  \centering
  \begin{subfigure}[b]{.33\linewidth}
    \centering
    \includegraphics[width=.5\linewidth]{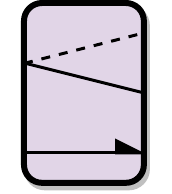}
    \caption{Switch/set-up line.}
    \label{fig:basis w/su}
  \end{subfigure}
\hspace{5mm}
  \begin{subfigure}[b]{.33\linewidth}
    \centering
    \includegraphics[width=.5\linewidth]{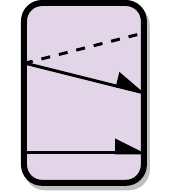}
    \caption{Set-up switch/set-up line.}
    \label{fig:basis suw/su}
  \end{subfigure}

  \caption{A basis for bounded multi-input gadgets:
    all such gadgets can simulate one of these two.}
  \label{fig:bounded basis}
\end{figure}

\begin{figure}
  \centering
  \begin{subfigure}[b]{.3\linewidth}
    \centering
    \includegraphics[width=.5\linewidth]{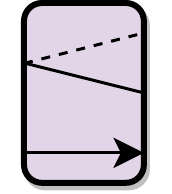}
    \caption{Switch/toggle line.}
    \label{fig:basis w/t}
  \end{subfigure}
  \begin{subfigure}[b]{.35\linewidth}
    \centering
    \includegraphics[width=.4\linewidth]{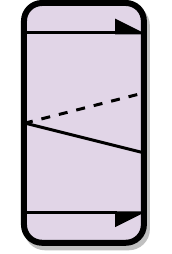}
    \caption{Switch/set-up line/set-down line}
    \label{fig:basis w/su/sd}
  \end{subfigure}
  \begin{subfigure}[b]{.3\linewidth}
    \centering
    \includegraphics[width=.5\linewidth]{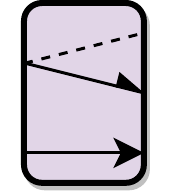}
    \caption{Set-up switch/toggle line.}
    \label{fig:basis suw/t}
  \end{subfigure}
  \begin{subfigure}[b]{.3\linewidth}
    \centering
    \includegraphics[width=.5\linewidth]{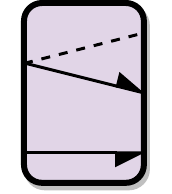}
    \caption{Set-up switch/set-down line.}
    \label{fig:basis suw/sd}
  \end{subfigure}
  \begin{subfigure}[b]{.35\linewidth}
    \centering
    \includegraphics[width=.43\linewidth]{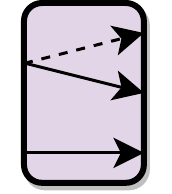}
    \caption{Toggle switch/toggle line.}
    \label{fig:basis tw/t}
  \end{subfigure}
  \begin{subfigure}[b]{.3\linewidth}
    \centering
    \includegraphics[width=.5\linewidth]{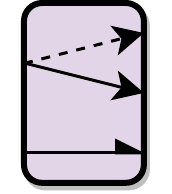}
    \caption{Toggle switch/set-up line.}
    \label{fig:basis tw/su}
  \end{subfigure}

  \caption{A basis for the unbounded multi-input gadgets:
    all such gadgets can simulate one of these six.
    We later show that Figure~\ref{fig:basis w/su/sd} alone forms a one-gadget basis.}
  \label{fig:unbounded basis}
\end{figure}

\begin{lemma}\label{lem:gadget basis}
  Let $G$ be an output-disjoint deterministic 2-state input/output gadget with multiple nontrivial inputs.
  \begin{itemize}
    \item If $G$ is bounded, then it simulates either a switch/set-up line or a set-up switch/set-up line (Figure~\ref{fig:bounded basis}).
    \item If $G$ is unbounded, then it simulates one of the following gadgets
    (Figure~\ref{fig:unbounded basis}):
    \begin{enumerate}[(a)]
      \item switch/toggle line,
      \item switch/set-up line/set-down line,
      \item set-up switch/toggle line,
      \item set-up switch/set-down line,
      \item toggle switch/toggle line, or
      \item toggle switch/set-up line.
    \end{enumerate}
  \end{itemize}
\end{lemma}

\begin{proof}
First we \emph{compress} every switch, set switch, and toggle switch,
except for one, by merging (connecting) its two outputs.
This operation transforms set switches into set lines,
toggle switches into toggle lines, and ordinary switches into trivial lines.
Figure~\ref{fig:SetUpSwitchSimulateSetUp} shows an example.
If the gadget has any ordinary switches, we use one of them as the switch
that does not get compressed.
The resulting gadget has the same boundedness as the original gadget,
has a single switch of some type, and still has multiple nontrivial inputs:
if it had only one nontrivial input, then the other inputs must have all been
ordinary switches which got compressed, so the remaining uncompressed input
is also an ordinary switch, and thus the original gadget contained only
ordinary switches and was trivial.

\begin{figure}[htbp]
  \centering
  \includegraphics[scale=.8]{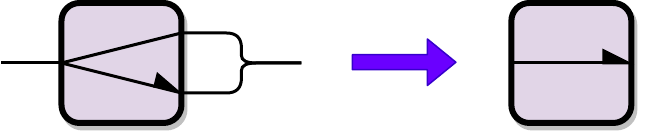}
  \caption{Compressing a set-up switch by merging its outputs yields a set-up line.}
  \label{fig:SetUpSwitchSimulateSetUp}
\end{figure}

For multi-input bounded gadgets, we now have either a switch or a set switch (any sort of toggle would make the gadget unbounded), and at least one set line. Each set switch and line must set the gadget to the same state (which we can assume by symmetry is the up state), and we can ignore all but one set line. In particular, without loss of generality, the resulting gadget contains exactly a set-up line and either a switch (\ref{fig:basis w/su}) or a set-up switch (\ref{fig:basis suw/su}).

For multi-input unbounded gadgets, there are multiple cases to consider based on the type of the single switch which was not compressed.
First, if the switch is an ordinary switch, then there must be lines that can set the state in both directions, which must include either a toggle line (\ref{fig:basis w/t}) or two set lines in different directions (\ref{fig:basis w/su/sd}).
If the switch is a set switch, then there must be a line that can set the state in the opposite direction, which can be either a toggle line (\ref{fig:basis suw/t}) or a set line opposite the set switch (\ref{fig:basis suw/sd}).
Finally, if the switch is a toggle switch, then there must be some nontrivial line: either a toggle line (\ref{fig:basis tw/t}) or a set line (\ref{fig:basis tw/su}). We have made arbitrary choices for the directions of set lines and set switches; these are without loss of generality because we can reflect the gadget (or rename the up and down states).
\end{proof}

These simulation results are of independent interest.
They show that there is a two-gadget \emph{basis}
for multi-input bounded input/output gadgets,
and a six-gadget basis for multi-input unbounded input/output gadgets,
where every gadget in each family can simulate at least one gadget
in the basis.
In fact, Section~\ref{sec:simulating sud} shows the stronger result
that multi-input unbounded input/output gadgets have a one-gadget basis,
namely, the switch/set-up line/set-down line of
Figure~\ref{fig:w/su/sd state diagram} or~\ref{fig:basis w/su/sd}.
Past work on one-player motion planning \cite{DHL}
established a one-gadget basis for a particular gadget family:
every reversible deterministic interacting-$k$-tunnel gadget
can simulate a locking 2-toggle.

At the other extreme from a basis, we can ask for \emph{universality}.
For example, in one-player motion planning,
each door gadget from \cite{Doors_FUN2020} simulates \textit{every} gadget.
In Section~\ref{sec:universality w/su/sd},
we prove a universality result for zero-player motion planning:
the same switch/set-up line/set-down line of
Figure~\ref{fig:w/su/sd state diagram} or~\ref{fig:basis w/su/sd}
simulates every deterministic input/output gadget
(not just those that are output-disjoint and 2-state).
Thus the switch/set-up line/set-down line
both simulates and can be simulated by
every unbounded multi-input output-disjoint
deterministic 2-state input/output gadget,
and thus every such gadget is similarly both a basis and universal.
We also prove a new universality result for one-player motion planning:
the switch/set-up line/set-down line---and
thus every unbounded multi-input output-disjoint
deterministic 2-state input/output gadget---simulates \textit{every} gadget
(just like the doors of \cite{Doors_FUN2020}).

\section{Zero Players}\label{sec:0p}

In this section, we study the complexity of zero-player motion planning with deterministic input/output gadgets from several classes. In Section~\ref{sec:0p 1input}, we consider such gadgets with a single input. In Section~\ref{sec:0p bounded}, we consider bounded gadgets with multiple inputs, which are naturally P-complete. Finally, in Section~\ref{sec:0p unbounded} we consider unbounded gadgets with multiple inputs, which are naturally PSPACE-complete.

\begin{lemma}\label{lem:0p in pspace}
  Zero-player motion planning with deterministic input/output gadgets is in PSPACE.
\end{lemma}

\begin{proof}
  In polynomial space, we can keep track of the current configuration of a system of gadgets and current location of the agent. Thus we can simply simulate the zero-player motion planning problem until either the agent reaches the goal location, the agent reaches a dead-end, or it makes more transitions than there are configurations, and thus is stuck in a cycle.
\end{proof}

\subsection{Single Input}\label{sec:0p 1input}

In this section, we consider zero-player motion planning with deterministic single-input input/\allowbreak output gadgets. If the gadgets are described (for concreteness, using transition graphs) as part of the instance, this is equivalent to the explicit zero-player reachability switching games of \cite{switchinggames}. In our language, \cite{switchinggames} shows that zero-player motion planning with instance-specified deterministic single-input input/output gadgets is NL-hard.
As pointed out in \cite{switchinggames}, the proofs in \cite{arrivalcls}, which only considered ARRIVAL, also apply to explicit zero-player reachability switching games. In our language, they show that zero-player motion planning with instance-specified deterministic single-input input/output gadgets is in UP~$\cap$~coUP (which is contained in NP~$\cap$~coNP).

We strengthen the NL-hardness result of \cite{switchinggames} by showing that zero-player motion planning with just the toggle switch is NL-hard. This is a straightforward modification of the proof of NL-hardness in \cite{switchinggames}; we present the full argument for completeness and to translate it to our terminology. There is still a large gap between the lower bound of NL-hard and the upper bound of UP~$\cap$~coUP.

\begin{theorem}
  Zero-player motion planning with the toggle switch is NL-hard.
\end{theorem}

\begin{proof}
  We reduce from reachability in directed graphs, which is NL-complete
  \cite{connectivity}.

  First we modify the graph to have out-degree $0$ or $2$ at every vertex
  without changing reachability; refer to Figure~\ref{fig:outdegree-2}.
  We replace every vertex $v$ with out-degree $k>2$ with a sequence of $k$ vertices each with out-degree at most~$2$: if $v$ has edges to $w_1,\dots,w_k$, we replace $v$ with $v_1,\dots,v_k$ with edges $v_i\to v_{i+1}$ and $v_i\to w_i$, and edges to $v$ now go to $v_1$.
  Then we remove any vertices with out-degree $1$ by setting their incoming edges to instead go to the target of their unique outgoing edge.
  This reduction to where every vertex has out-degree exactly $2$
  can be done in logarithmic space and does not affect reachability,

  \begin{figure}[htbp]
    \centering
    \includegraphics[width=\linewidth]{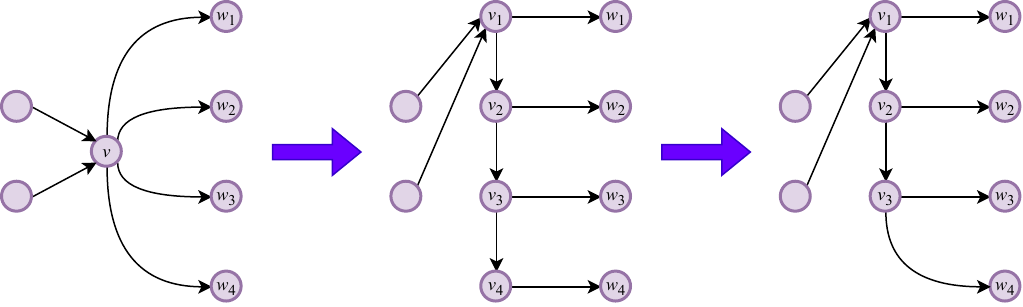}
    \caption{Modifying a directed graph to have out-degree $0$ or~$2$
      at every vertex: splitting vertices of high out-degree, and
      removing vertices of out-degree $1$.}
    \label{fig:outdegree-2}
  \end{figure}

  Now we use a construction based on that in \cite{switchinggames};
  refer to Figure~\ref{fig:toggle-switch-NL}.
  Let $V$ be the set of vertices in the modified graph $G$,
  where we are interested in a path from $s$ to $t$.
  Our system of gadgets has $|V|^2$ toggle switches,
  named $(v,i)$ for $v\in V$ and $1\leq i\leq|V|$.
  For a vertex $v\neq t$ with edges to $w_1$ and $w_2$ and $i<|V|$,
  the outputs of $(v,i)$ are connected to the inputs of $(w_1,i+1)$ and
  $(w_2,i+1)$.
  For a vertex $v\neq t$ with out-degree $0$,
  both outputs of $(v,i)$ are connected to the input of $(s,1)$.
  For $v\neq t$, both outputs of $(v,|V|)$ are connected
  to the input of $(s,1)$.
  Finally, for each $i$, both outputs of $(t,i)$ are connected to the
  goal location, which then leads back to $(s,1)$.
  The start location is the input of $(s,1)$.

  \begin{figure}[htbp]
    \centering
    \includegraphics[scale=0.5]{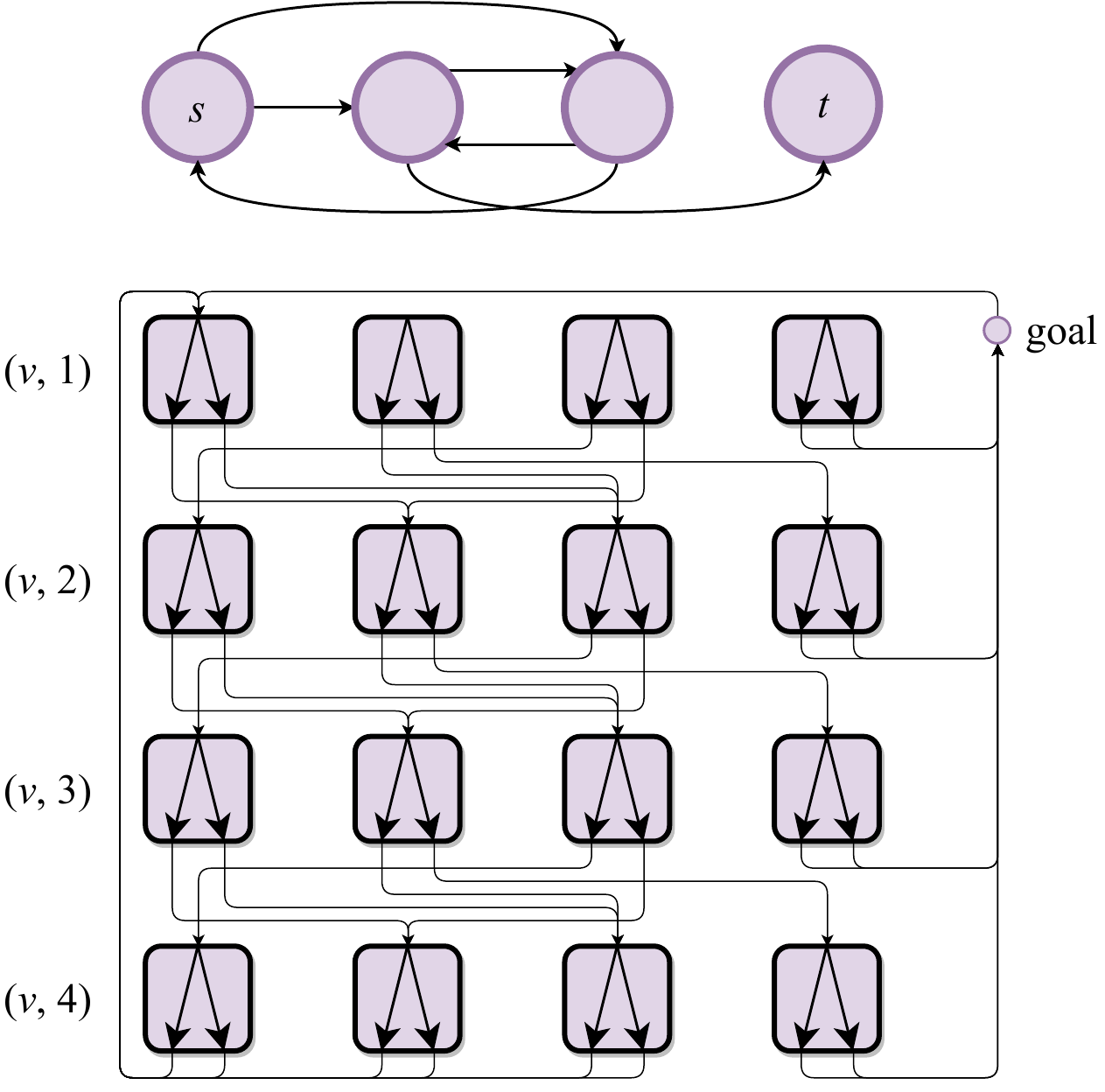}
    \caption{Reduction from reachability in directed graphs of out-degree $0$ or $2$ to zero-player motion planning with the toggle switch.}
    \label{fig:toggle-switch-NL}
  \end{figure}

  When the agent moves through this system, it follows paths in $G$ starting from $s$ and counts the number of steps taken, resetting after $|V|$ steps
  or when it reaches a vertex with out-degree $0$.
  By construction, a toggle switch $(v,i)$ is reachable from $(s,1)$
  exactly when there is a path of length $i-1$ from $s$ to $v$.
  If the agent reaches the goal location, it must have entered $(t,i)$ for some $i$, and thus there is a path (of length $i-1$) from $s$ to $t$.

  Because all paths eventually return to $(s,1)$, the agent must enter $(s,1)$ infinitely many times, so it must use each output of $(s,1)$ infinitely many times.
  By induction, it uses every toggle switch reachable from $(s,1)$ infinitely many times.
  If there is a (simple) path from $s$ to $t$, it has some length $i<|V|$, so $(t,i+1)$ is reachable from $(s,1)$.
  Then $(t,i+1)$ is visited infinitely many times, so the agent reaches the goal location.
\end{proof}

\subsection{Bounded Gadgets}\label{sec:0p bounded}
In this section, we consider the complexity of zero-player motion planning with a bounded output-disjoint deterministic 2-state input/output gadget which has multiple nontrivial inputs. We will find that this problem is always P-complete.

A gadget is \emph{bounded} if the number of times it can change states is bounded; this generalizes the definition in Section~\ref{sec:intro characterization}.

\begin{theorem}\label{thm:bounded in P}
  Zero-player motion planning with bounded deterministic input/output gadgets is in~P.
\end{theorem}

\begin{proof}
  Suppose we have a system with $n$ copies of the gadget.
  Let $k$ be the maximum number of state changes a gadget in the system can make,
  and let $i$ be the maximum number of input locations a gadget in the system has.
  Then gadget states can change at most $k n$ times.
  Between consecutive state changes, the agent can visit each entrance of each gadget at most once
  (otherwise it is stuck in a cycle),
  so consecutive state changes are separated by at most $i n$ traversals.
  Hence after $i k n^2$ traversals, the agent must be in a cycle,
  which involves no state changes of length at most $i n$.
  So we can solve the problem in polynomial time by simulating the agent for $i n (k n+1)$ steps and seeing whether it reaches the goal location by then.
\end{proof}

Lemma~\ref{lem:gadget basis} tells us that every output-disjoint deterministic 2-state input/output gadget with multiple nontrivial inputs simulates either the switch/set-up line or the set-up switch/set-up line. Thus to prove that zero-player motion planning with any such gadget is P-hard, it suffices to show P-hardness for these particular two gadgets. This is what we do for the remainder of this section.

\begin{theorem}
  Zero-player motion planning with the switch/set-up line or the set-up switch/set-up line is P-hard (under logarithmic space reductions).
\end{theorem}

\begin{proof}
  We provide a reduction to each of these problems from the problem of evaluating a circuit containing only $\NOR$ gates and fan-out, with the gates listed in a topological order.
  This restricted version of circuit evaluation is known to be P-complete \cite{greenlaw1995limits}. The two reductions are nearly identical: we present the reduction for the switch/set-up line, and the reduction for the set-up switch/set-up line is the same with each gadget replaced. We shall see that the agent never goes over a switch multiple times, so these two systems of gadgets behave the same.

  Our reduction builds a system of switch/set-up lines which has one gadget for each input of a $\NOR$ gate; this gadget indicates whether the input is true or false, and is initially set to false. The agent will evaluate each $\NOR$ gate in the order they are listed in the input,
  setting the gadgets for outputs of that gate to true if appropriate. This is accomplished with the gadget in Figure~\ref{fig:circuiteval}. For each $\NOR$ gate, we build one of these gadgets, where $x$ and $y$ are the inputs, and the gadgets labeled $x \NOR y$ are the outputs (and inputs of other $\NOR$ gates). There are as many output gadgets as the fan-out of this $\NOR$ gate. The entrance and exit to the $\NOR$ gate gadgets are connected in series, in the given order of the gates.

  \begin{figure}
    \centering
    \includegraphics[scale=.8]{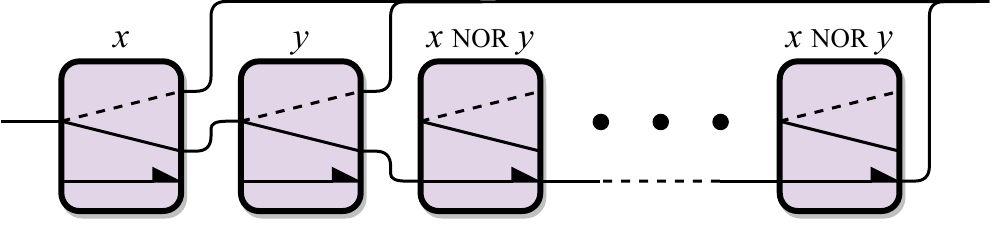}
    \caption{A $\NOR$ gate for P-hardness of zero-player motion planning with the switch/set-up line. If neither $x$ nor $y$ is set to true (up), the agent sets each $x \NOR y$ gadget to true.}
    \label{fig:circuiteval}
  \end{figure}

  To complete the construction, we place the start location at the entrance to the first $\NOR$ gate. The exit of the last $\NOR$ gate enters a switch which holds the output of the final $\NOR$ gate, and the goal location is the top output of that switch. Every switch/set-up line starts in the down state except for those that correspond to true inputs to the circuit.

  When the agent moves through this system of gadgets, in goes through each $\NOR$ gate in order.
  Because the input circuit was given with gates in topological order, the agent goes through both gates that provide the inputs $x$ and $y$ to a gate that computes $x \NOR y$ before going through that gate itself.
  If either $x$ or $y$ is set to true (i.e., in the up state), the agent leaves $x \NOR y$ false, but if $x$ and $y$ are both false, it goes through the set-up lines to set $x \NOR y$ true. This correctly computes $x \NOR y$, and by induction it computes the value of the circuit. At the end, the agent reaches the goal location if the value is true and gets stuck in a nearby dead-end if the value is false.
\end{proof}

By the basis simulation result of Lemma~\ref{lem:gadget basis},
these two cases establish hardness for
all multi-input output-disjoint deterministic 2-state input/output gadgets:

\begin{corollary}
  Zero-player motion planning with any bounded output-disjoint deterministic 2-state input/output gadget with multiple nontrivial inputs is P-complete.
\end{corollary}

\subsection{Unbounded Gadgets}\label{sec:0p unbounded}
In this section, we consider zero-player motion planning with an unbounded output-disjoint deterministic 2-state input/output gadget which has multiple nontrivial inputs. We show that this problem is PSPACE-complete for every such gadget through a reduction from Quantified Boolean Formula (QBF), which is PSPACE-complete, to zero-player motion planning with the switch/set-up line/set-down line, and by showing that every such gadget simulates the switch/set-up line/set-down line. We also show that the switch/set-up line/set-down line (and thus every unbounded output-disjoint deterministic 2-state input/output gadget with multiple nontrivial inputs) can simulate every deterministic input/output gadget in zero-player motion planning.

\subsubsection{Edge Duplicators}

Many of our simulations involve building an \emph{edge duplicator}, shown in Figure~\ref{fig:edge duplicator}. An edge duplicator is a construction which allows us to effectively make a copy of a line from $X$ to $X^\prime$ in a gadget. This is achieved by routing two inputs $A$ and $B$ to $X$, and then sending the agent from $X^\prime$ to one of two exits $A^\prime$ or $B^\prime$ corresponding to the input used. The details of the construction of an edge duplicator depend on the gadget used; see Figure~\ref{fig:su/w/sd duplicator} for an example.

If we have access to an edge duplicator, we can duplicate tunnels in gadgets. Note that this is not enough to duplicate switches, since we would have to account for both exits getting duplicated.

\begin{figure}
  \centering
  \includegraphics[scale=.8]{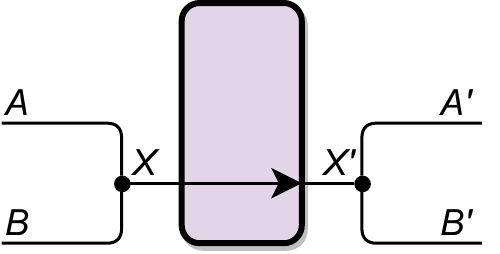}
  \caption{The schematic of an edge duplicator. An agent entering at $A$ or $B$ exits at $A^\prime$ or $B^\prime$, respectively, having gone over the central path. This duplicates the edge in the center.}
  \label{fig:edge duplicator}
\end{figure}

\subsubsection{PSPACE-Hardness of the Switch/Set-Up Line/Set-Down Line}

In this section, we show that zero-player motion planning with the switch/set-up line/set-down line is PSPACE-hard through a reduction from QBF.
Recall from Figure~\ref{fig:w/su/sd state diagram} or~\ref{fig:basis w/su/sd}
that the switch/set-up line/set-down line is a 2-state input/output gadget
with three inputs: one sets the state to up, one sets it to down,
and one sends the agent to one of two outputs based on the current state.

\begin{theorem}
  Zero-player motion planning with the switch/set-up line/set-down line is PSPACE-hard.
\end{theorem}
\label{thm:SUDPSPACE}

\begin{proof}
  First we build an edge duplicator, shown in Figure~\ref{fig:su/w/sd duplicator}. This allows us to use gadgets with multiple set-up or set-down lines.

  \begin{figure}
    \centering
    \includegraphics[scale=.8]{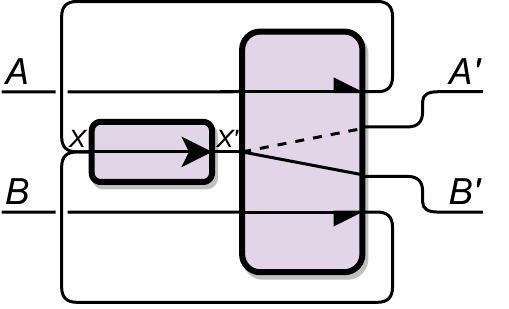}
    \caption{An edge duplicator for the switch/set-up line/set-down line. A agent entering on the left sets the state of the switch, goes across the duplicated tunnel (which is one subunit of some bigger gadget), and exits based on the state it set the switch to.}
    \label{fig:su/w/sd duplicator}
  \end{figure}

  Next we present a reduction from QBF. Given a quantified Boolean formula where the unquantified formula is 3-CNF, we construct a system of gadgets which evaluates the formula, ultimately sending the agent to one of two locations based on its truth value. The system consists of a sequence of \emph{quantifier gadgets}, which set the values of variables, followed by the \emph{CNF evaluation}, which checks whether the formula is satisfied by a particular assignment and reports this to the quantifier gadgets.

  Each quantifier gadget has three inputs, called In, True-In, and False-In, and three outputs, called Out, True-Out, and False-Out. The agent will always first arrive at In. This sets the variable controlled by that quantifier to true, and the agent leaves at Out, which sends it to the next quantifier gadget. Eventually the agent will return to either True-In or False-In, depending on the truth value of the rest of the quantified formula with the variable set to true. Depending on the result, the quantifier gadget either sends the agent to True-Out or False-Out to pass this message to the previous quantifier gadget, or the quantifier gadget sets its variable to false, and again sends the agent to the next quantifier. When it gets a truth value in response the second time, it sends the appropriate truth value to the previous quantifier.
  The last quantifier communicates with the CNF evaluation instead of with another quantifier.

  The universal quantifier gadget is shown in Figure~\ref{fig:quantifier}. The chain of gadgets at the top encode the state of the variable controlled by this quantifier, as has as many gadgets as there are instances of the variable in the formula. The variable is true when they are set to the `left' state and false when they are set to the `right' state, where the direction refers to the position, in the figure as drawn, of the exit which would be taken if the agent enters the switch.

  When the agent enters In, it sets the variable to true and exits Out. If it then returns to True-In, the first time it takes the bottom branch of the switch, sets that gadget to the up state, sets the variable to false, and exits Out again. If it returns to True-In a second time, that means the rest of the formula was true for both settings of the universally quantified variable: it takes the top branch, resets that gadget to down, and exits True-Out. If after either trial the agent enters at False-In, it resets the bottom gadget to the down state and exits False-Out. This is the intended behavior of the universal quantifier: it reports true if the result was true for both settings of the variable, and false otherwise.

  \begin{figure}
    \centering
    \includegraphics[scale=.8]{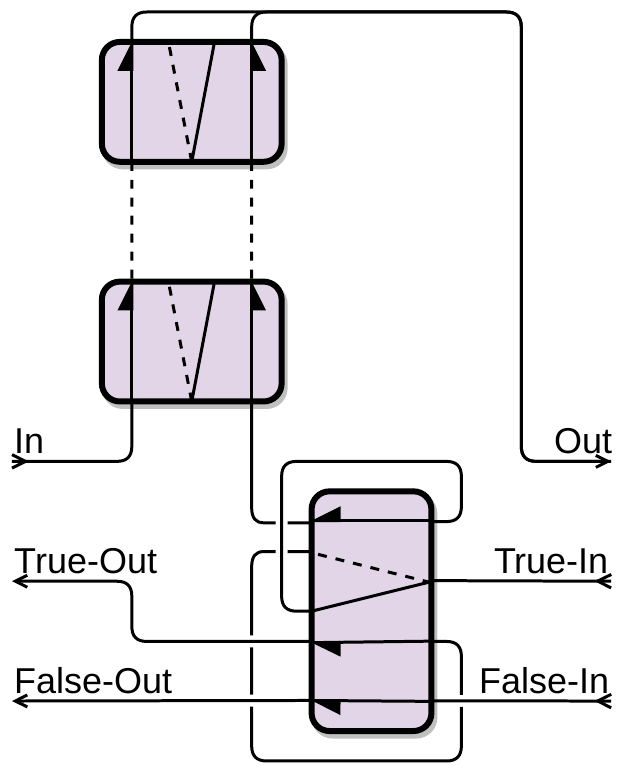}
    \caption{The universal quantifier for the switch/set-up line/set-down line. An edge duplicator (Figure~\ref{fig:su/w/sd duplicator}) is used to give the bottom gadget two set-down lines.}
    \label{fig:quantifier}
  \end{figure}

  The existential quantifier is identical except that True-Out and False-Out are swapped, and True-In and False-In are swapped. It reports false if the result was false for both settings, and true otherwise.

  For CNF evaluation, we use the switches controlled by each quantifier to read the value of a variable. For each clause, the agent passes through a switch corresponding to each of the literals in the clause. If all three literals are false, it exits False-Out. Otherwise, it moves on to the next clause, eventually exiting True-Out if all clauses are satisfied. This is shown, for 3 clauses, in Figure~\ref{fig:cnf eval}. Ultimately, the agent exits True-Out or False-Out depending on whether the formula is satisfied by the current assignment.

  \begin{figure}
    \centering
    \includegraphics[scale=.8]{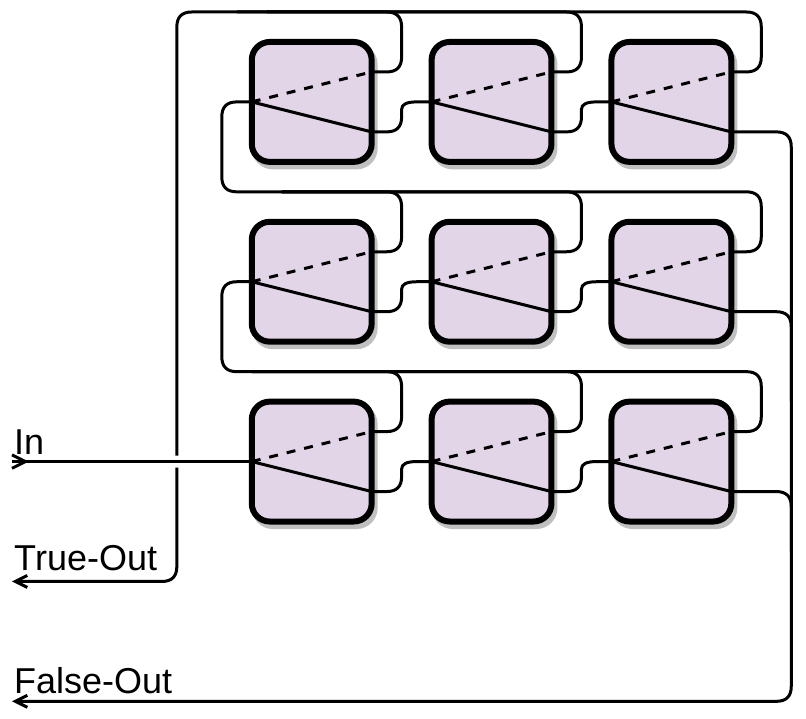}
    \caption{Three clauses of CNF evaluation for the switch/set-up line/set-down line; each clause is a row of three switches. The switches are part of gadgets in the quantifiers. We assume the top exit of each switch corresponds to that literal being true; all literals are set to false in this image.}
    \label{fig:cnf eval}
  \end{figure}

  It follows by induction that, for each quantifier, when the agent arrives at In, it will eventually leave either True-Out or False-Out depending on the truth value of the suffix of the formula beginning with that quantifier under the assignment of the earlier quantifiers. Thus, if the agent starts in the first quantifier at In, it reaches True-Out on the first quantifier if and only if the formula is true.
\end{proof}

\subsubsection{Other Gadgets Simulate the Switch/Set-Up Line/Set-Down line}\label{sec:simulating sud}

In this section, we show that every unbounded output-disjoint deterministic 2-state input/output gadget with multiple nontrivial inputs simulates the switch/set-up line/set-down line.
In other words, the switch/set-up line/set-down line forms a one-gadget basis for unbounded output-disjoint deterministic 2-state input/output gadgets.
We only need to show that the five other gadgets from Lemma~\ref{lem:gadget basis} simulate the switch/set-up/set-down. It follows that zero-player motion planning with any such gadget is PSPACE-complete, since we can replace each gadget in a system of switch/set-up/set-down with a simulation of it.

\paragraph{Toggle Switch/Toggle Switch.}

We begin with the toggle switch/toggle switch, which is not part of our basis of gadgets from Lemma~\ref{lem:gadget basis}, but will be a useful intermediate gadget.
It simulates an edge duplicator, as shown in Figure~\ref{fig:tw/tw duplicator}.
We can merge the two outputs of one of the toggle switches to simulate a toggle switch/toggle line, and then duplicate the toggle line to make a gadget with one toggle switch and any number of toggle lines.

\begin{figure}
\centering
\begin{minipage}{0.47\linewidth}
  \centering
  \includegraphics[scale=.8]{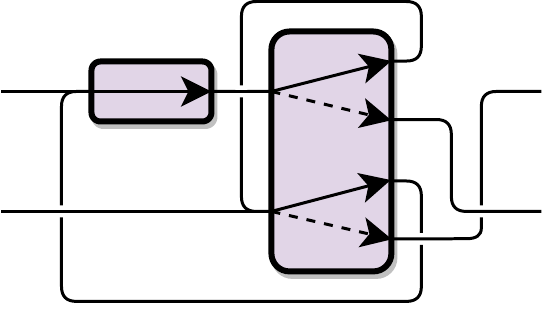}
  \caption{An edge duplicator for the toggle switch/toggle switch. The tunnel on the left is duplicated.}
  \label{fig:tw/tw duplicator}
\end{minipage}\hfill
\begin{minipage}{0.5\linewidth}
  \centering
  \includegraphics[scale=.8]{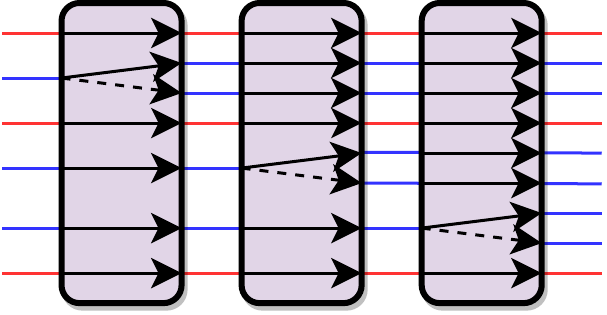}
  \caption{A simulation of three toggle lines and three toggle switches from gadgets with one toggle switch and 5, 6, and 7 toggle lines. The red tunnels are toggle lines and the blue tunnels are toggle switches.}
  \label{fig:t^n/tw^n}
\end{minipage}
\end{figure}

By putting such gadgets in series, we can simulate a gadget with any number of toggle lines and any number of toggle switches. Figure~\ref{fig:t^n/tw^n} shows this for three toggle lines and three toggle switches, which is as large as we need. This simulated gadget can finally simulate the switch/set-up line/set-down line, as shown in Figure~\ref{fig:tw/tw -> w/su/sd}.

\begin{figure}
\centering
\begin{minipage}{0.5\linewidth}
  \centering
  \includegraphics[scale=.8]{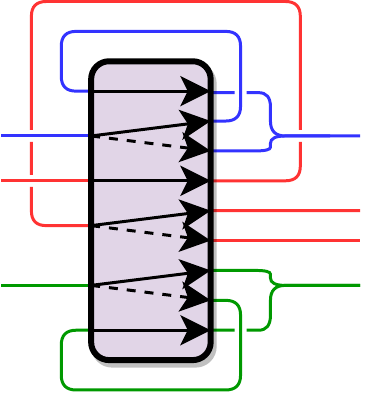}
  \caption{A simulation of a switch/set-up line/\allowbreak set-down line, currently in the down state, from the gadget built in Figure~\ref{fig:t^n/tw^n}. Each component of the switch/set-up line/set-down line is made from one toggle line and one toggle switch; the switch, set-up line, and set-down line are red, green, and blue, respectively.}
  \label{fig:tw/tw -> w/su/sd}
\end{minipage}\hfill
\begin{minipage}{0.47\linewidth}
  \centering
  \includegraphics[scale=.8]{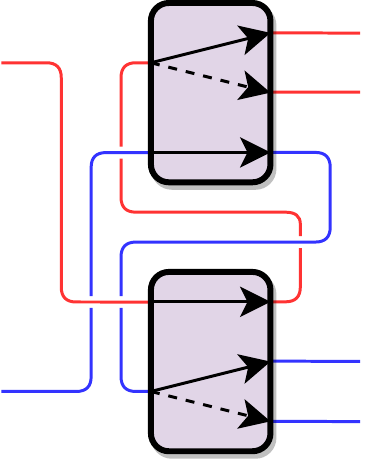}
  \caption{A simulation of a toggle switch/\allowbreak toggle switch from the toggle switch/\allowbreak toggle line. Each color corresponds to one of the toggle switches.}
  \label{fig:tw/t -> tw/tw}
\end{minipage}
\end{figure}

\paragraph{Toggle Switch/Toggle Line.}

We simulate the toggle switch/toggle switch using toggle switch/\allowbreak toggle lines, as shown in Figure~\ref{fig:tw/t -> tw/tw}.

\paragraph{Switch/Toggle Line.}

First we build an edge duplicator, shown in Figure~\ref{fig:w/t duplicator}.
Then we can duplicate the toggle line and put one copy in series with the switch, constructing a toggle switch/toggle line.

\begin{figure}
\centering
\begin{minipage}{0.55\linewidth}
  \centering
  \includegraphics[scale=.8]{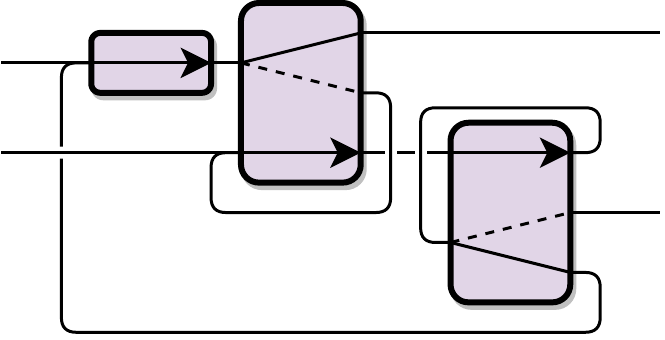}
  \caption{An edge duplicator for the switch/toggle line. The leftmost tunnel is duplicated.}
  \label{fig:w/t duplicator}
\end{minipage}\hfill
\begin{minipage}{0.42\linewidth}
  \centering
  \includegraphics[scale=.8]{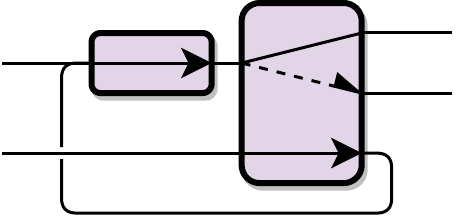}
  \caption{An edge duplicator for the set-up switch/toggle line. The leftmost tunnel is duplicated.}
  \label{fig:suw/t duplicator}
\end{minipage}
\end{figure}

\paragraph{Set-Up Switch/Toggle Line.}

First we build an edge duplicator, shown in Figure~\ref{fig:suw/t duplicator}.
Then we simulate the switch/toggle line, shown in Figure~\ref{fig:suw/t -> w/t}.

\begin{figure}
\centering
\begin{minipage}{0.3\linewidth}
  \centering
  \includegraphics[scale=.8]{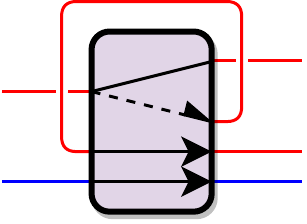}
  \caption{A simulation of the switch/toggle line using the set-up switch/toggle line. Red corresponds to the switch and blue
  corresponds to the toggle line.}
  \label{fig:suw/t -> w/t}
\end{minipage}\hfill
\begin{minipage}{0.67\linewidth}
  \centering
  \includegraphics[scale=.8]{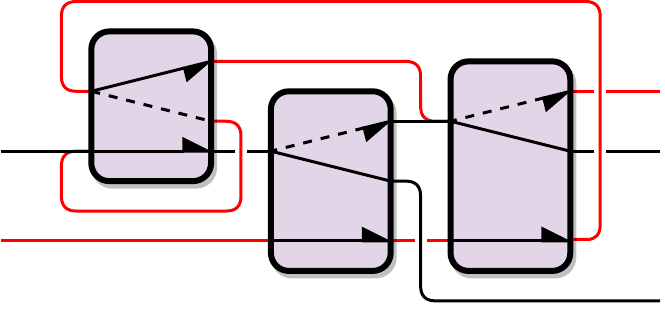}
  \caption{A simulation of a set-down switch/toggle line using the set-up switch/set-down line. When the agent is not inside the simulation, the rightmost gadget is in the down state and the other two gadgets are in opposite states and encode the state of the simulated gadget. Red lines indicate the toggle line: when the agent enters the bottom entrance, it takes one of the internal paths depending on the state and exits the top exit, reversing the state of the left and middle gadgets. When it enters the top entrance, it exits one of the bottom two exits and resets the state to down.}
  \label{fig:suw/sd -> suw/t}
\end{minipage}
\end{figure}

\paragraph{Set-Up Switch/Set-Down Line.}

We simulate a set-down switch/toggle line (equivalent to a set-up switch/toggle line) using the set-up switch/set-down line, as shown in Figure~\ref{fig:suw/sd -> suw/t}.

\paragraph{Toggle Switch/Set-Up Line.}

We simulate a set-up line/set-down switch using the toggle switch/\allowbreak set-up line, as shown in Figure~\ref{fig:su/tw -> su/sdw}; this is equivalent to a set-up switch/set-down line.

\begin{figure}
  \centering
  \includegraphics[scale=.8]{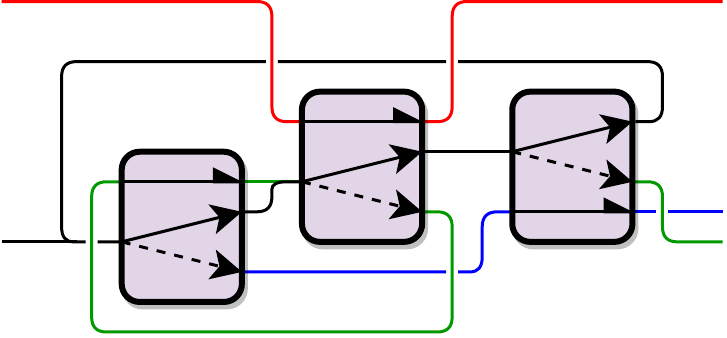}
  \caption{A simulation of a set-up line/set-down switch from the set-up line/toggle switch. The state of the simulated gadget is the same as the state of the center gadget. The red path corresponds to the set-up line. When it enters the set-down switch, the agent goes along the blue lines if the state is down, the green lines if the state is up, and the black lines in both cases.}
  \label{fig:su/tw -> su/sdw}
\end{figure}

\medskip

These simulations, together with Lemma~\ref{lem:gadget basis}, give the following theorem:

\begin{theorem}
\label{thm:s/su/sd simulation}
  Every unbounded output-disjoint deterministic 2-state input/output gadget with multiple nontrivial inputs simulates the switch/set-up line/set-down line.
\end{theorem}

\begin{corollary} \label{cor:0p pspace-complete}
  Let $G$ by an unbounded output-disjoint deterministic 2-state input/output gadget with multiple nontrivial inputs. Then zero-player motion planning with $G$ is PSPACE-complete.
\end{corollary}

\begin{proof}
  Containment in PSPACE is given by Lemma~\ref{lem:0p in pspace}. All of our simulations preserve PSPACE-hardness: we can reduce from zero-player motion planning with the switch/set-up line/set-down line (shown PSPACE-hard in Theorem~\ref{thm:SUDPSPACE}) to zero-player motion planning with $G$ by replacing each gadget in a system of switch/set-up line/set-down lines with a simulation built from $G$. The resulting system of $G$ has the same behavior as the system of switch/set-up line/set-down lines.
\end{proof}

\subsubsection{Universality of the Switch/Set-Up Line/Set-Down Line}
\label{sec:universality w/su/sd}

In this section, we show how to simulate an arbitrary deterministic input/output
gadget using the switch/set-up line/set-down line, i.e., that this gadget is
universal for all deterministic input/output gadgets.
We also show interesting consequences of this result; of particular note is
Corollary~\ref{cor:w/su/sd universal 1p} that, in one-player motion planning,
the switch/set-up line/set-down line simulates every gadget, i.e.,
is fully universal (just like the doors of \cite{Doors_FUN2020}).

\begin{theorem}\label{thm:w/su/sd universal 0p}
  The switch/set-up line/set-down line simulates every deterministic input/output gadget in zero-player motion planning.
\end{theorem}

\begin{proof}
  We present simulations of gradually more powerful gadgets. First, the edge duplicator (Figure~\ref{fig:su/w/sd duplicator}) lets us have any number of copies of the set-up and set-down lines.

  Next, we simulate a generalization of the switch/set-up line/set-down line which call the $k$-switch. This gadget has $k$ states, $k$ lines which each set the gadget to a particular state, and an input which does not change the state and sends the agent to one of $k$ locations depending on the state. The switch/set-up line/set-down line is a 2-switch. The simulation for $k=4$ is shown in Figure~\ref{fig:w/su/sd -> kw}, and generalizes easily to arbitrary $k$: we need $k-1$ gadgets connected in series, where the $i$th gadget has $i$ set-up lines and $k-1-i$ set-down lines.

  We now duplicate the large switch in a $k$-switch using the construction in Figure~\ref{fig:w dup}. Thus the switch/set-up line/set-down line can simulate a gadget with any number of states, any number of lines which set it to a particular state, and any number of inputs which send the agent to different outputs depending on the state but do not change the state.

  Finally, let $G$ be an arbitrary deterministic input/output gadget. If $G$ has $k$ states and $m$ input locations, we use a $k$-switch with $m$ copies of the switch to simulate $G$. The $m$ inputs lead directly to the $m$ switches. For each transition $(q,a)\to(r,b)$ of $G$---meaning that when the agent enters at $a$ in state $q$, it exits at $b$ and changes the state to~$r$---we connect the output taken in $q$ of the switch corresponding to $a$ to a line which sets the state to $r$, and connect the output of that line to $b$. This encodes the correct behavior for that transition. Because $G$ is deterministic, there is only one such transition for each pair $(q,a)$, so only connect each output of a switch to one input location, as required for zero-player motion planning.
\end{proof}

  \begin{figure}
    \centering
    \includegraphics[scale=.8]{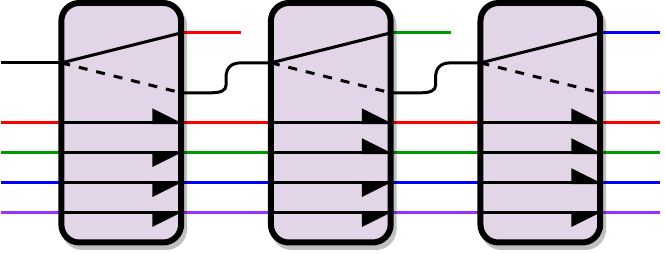}
    \caption{A simulation of a 4-switch using the switch/set-up line/set-down line. Colors indicate the outputs corresponding to set lines.}
    \label{fig:w/su/sd -> kw}
  \end{figure}

  \begin{figure}
    \centering
    \includegraphics[scale=.8]{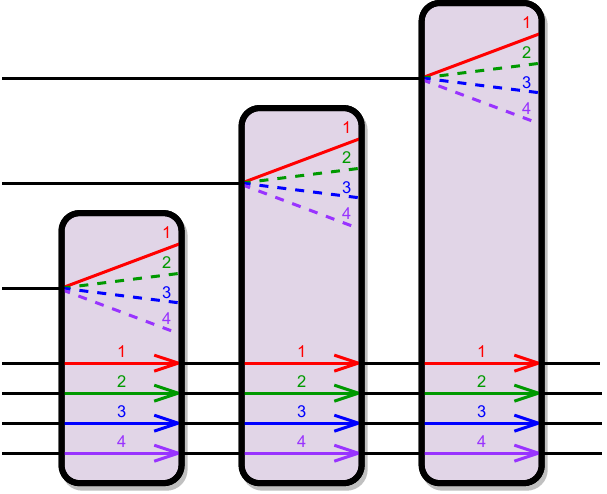}
    \caption{Simulating a 4-switch which has three copies of the switch.}
    \label{fig:w dup}
  \end{figure}

\begin{corollary}
  Every unbounded output-disjoint deterministic 2-state input/output gadget with multiple nontrivial inputs simulates every deterministic input/output gadget in zero-player motion planning.
\end{corollary}

\begin{corollary}\label{cor:w/su/sd io universal 1p}
  The switch/set-up line/set-down line simulates every input/output gadget in one-player motion planning (that is, we allow multiple input locations in the same connected component of the connection graph, or equivalently allow fan-out gadgets as described in Section~\ref{sec:1p}).
\end{corollary}

\begin{proof}
  We use the same construction as in the proof of Theorem~\ref{thm:w/su/sd universal 0p}. If $G$ is nondeterministic---say it has multiple transitions when entering $a$ in state $q$---then we will connect the output taken in $q$ of the switch corresponding to $a$ to multiple input locations.
\end{proof}

\begin{corollary}\label{cor:w/su/sd universal 1p}
  In one-player motion planning, the switch/set-up line/set-down line simulates every gadget.
\end{corollary}

\begin{proof}
  Let $G$ be an arbitrary gadget. We construct a gadget $G^\prime$ with the same states as $G$, locations $a_{in}$ and $a_{out}$ for each location $a$ of $G$, and a transition $(q,a_{in})\to(r,b_{out})$ for each transition $(q,a)\to(r,b)$ of $G$. Clearly $G^\prime$ is input/output: $a_{in}$ and $a_{out}$ are input and output locations, respectively. Thus, by Corollary~\ref{cor:w/su/sd io universal 1p}, the switch/set-up line/set-down line simulates $G^\prime$ in one-player motion planning. But $G^\prime$ simulates $G$ simply by connecting both $a_{in}$ and $a_{out}$ to~$a$.
\end{proof}

\begin{corollary}
  In one-player motion planning, every unbounded output-disjoint deterministic 2-state input/output gadget with multiple nontrivial inputs simulates every gadget.
\end{corollary}

\section{One Player}\label{sec:1p}
In this section, we consider one-player motion planning with input/output gadgets. This is a generalization of zero-player motion planning, where we no longer require each connected component of the connection graph to have only one input location. We also now allow nondeterministic gadgets.

A simple nondeterministic input/output gadget is the \emph{fan-out} gadget, which has one input location, two output locations, and one state; the player may choose which output location to take. One-player motion planning (with input/output gadgets) can be equivalently defined by introducing the fan-out gadget to zero-player motion planning, instead of removing the constraint that the system is branchless.

We characterize the complexity of one-player motion planning with an
output-disjoint deterministic 2-state input/output gadget as follows;
refer to the bottom row of Table~\ref{tbl:2-state}
and the middle row of Table~\ref{tbl:1 input}.
If the gadget is trivial (with at least one traversal), then one-player motion planning is just reachability in a directed graph, which is NL-complete \cite{connectivity}.
If the gadget is unbounded and multi-input,
then one-player motion planning is PSPACE-complete
by Lemma~\ref{lem:1p easy containment} below and by Corollary~\ref{cor:0p pspace-complete}
because it is a generalization of zero-player motion planning.
Otherwise, one-player motion planning is in NP
by Lemma~\ref{lem:1p easy containment} if it is bounded
and by Theorem~\ref{thm:1p in np} if it is single-input.
In either case, motion planning is NP-hard by Corollary~\ref{cor:1p np hard}.

We begin with straightforward containment results:

\begin{lemma}\label{lem:1p easy containment}
  One-player motion planning with input/output gadgets is in PSPACE,
  and one-player motion planning with bounded input/output gadgets is in NP.
\end{lemma}

\begin{proof}
  One-player motion planning can easily be simulated by a nondeterministic polynomial-space algorithm
  which guesses player choices,
  so it is in NPSPACE${}={}$PSPACE \cite{SAVITCH1970177}.

  For bounded input/output gadgets, define $k$, $i$, and $n$ as in Theorem~\ref{thm:bounded in P}.
  As before, the number of state changes is bounded by $kn$.
  The shortest solution visits each entrance at most once between consecutive state changes,
  and thus has total length at most $in(kn+1)$.
  This is a polynomial, so we can use the list of transitions as a certificate for NP.
\end{proof}

For the remainder of this section, we focus on one-player motion planning with single-input input/output gadgets.

One-player reachability switching games, studied in \cite{switchinggames}, are
equivalent to one-player motion planning with deterministic single-input
input/output gadgets. Fearnley, Gairing, Mnich, and Savini \cite{switchinggames}
show that this problem is NP-complete when the gadgets are described as part of
the instance.

In this section, we improve on this result in two ways. First, we show in Section~\ref{thm:1p in np} that the problem remains in NP even when we allow nondeterministic single-input input/output gadgets, which cannot all obviously be simulated by deterministic gadgets. Our proof is similar to the proof of containment in NP in \cite{switchinggames}.

Second, we show in Section~\ref{sec:1p np-hard} that the problem remains NP-hard with a specific gadget instead of instance-specified gadgets. In particular, we show that one-player motion planning with the toggle switch or the set switch is NP-complete. Our reduction is simpler than the one in \cite{switchinggames}, and the technique can be used to prove NP-hardness for many other single-input input/output gadgets.

\subsection{Containment in NP}\label{sec:1p in np}
First we show that one-player motion planning with any single-input input/output gadget is in NP, generalizing a result from \cite{switchinggames}. Our proof is similar, but requires more care to account for nondeterministic gadgets.

\begin{theorem}\label{thm:1p in np}
  One-player motion planning with single-input input/output gadgets is in NP.
\end{theorem}

The input can describe the gadgets in the system by listing their states and locations and specifying their transition graphs.

\begin{proof}
  A single-input input/output gadget is described by a labeled directed graph, with states as vertices and transitions as edges, where each edge is labeled with an output location. An edge labeled $b$ from $q$ to $r$ indicates that, when the agent enters the unique input location in state $q$, it can exit at $b$ and change the state to $r$. If you prefer, this can be thought of as a nondeterministic finite automaton (NFA) whose alphabet is the locations of the gadgets in the motion-planning problem.

  We will adapt the certificates used in \cite{switchinggames}, controlled switching flows, to work for nondeterministic gadgets. The number of times each output location (or edge in the equivalent reachability switching game) is used is no longer enough information, since it may in general be hard to determine whether a nondeterministic gadget has a legal sequence of transitions which uses each location a specified number of times.%
  \footnote{In fact, this is NP-hard by a reduction from the existence of a Hamiltonian path in a directed graph: given a graph with $n$ vertices, construct a gadget with $n$ states and $n$ output locations whose transition graph is the input graph, and ask for a sequence of transitions which uses each output location exactly once. In terms of finite automata, determining whether a given NFA accepts any anagram of a given string is NP-complete.}
  Instead, we will have the certificate include the number of times each \emph{traversal} in each gadget is used, which will be enough information to be checked quickly. We modify the definition of controlled switching flows as follows.

  \begin{definition}
    A \emph{controlled switching flow} in a system of single-input input/output gadgets is a function $f$ from the set of transitions gadgets in the system to the natural numbers (including zero) which is ``locally consistent'' in the following sense:
    \begin{itemize}
      \item For a connected component $H$ of the connection graph, let $H_i$ and $H_o$ be the sets of traversals from input locations and to output locations in $H$, respectively. That is, $H_i$ contains all transitions in gadgets whose input location is in $H$, and $H_o$ contains the transitions which leave the agent in $H$. Then
      $$\sum_{t\in H_i}f(t)-\sum_{t\in H_o}f(t)=
      \begin{cases}
        1 & H\text{ contains the start location}\\
        -1 & H\text{ contains the goal location\footnotemark}\\
        0 & \text{otherwise.}
      \end{cases}$$
      \footnotetext{We can assume the start and goal locations are in different connected components, since otherwise the reachability problem is trivial.}
      \item For each gadget, there is a legal sequence of transitions from its starting state $s$ which uses each transition $t$ in the gadget exactly $f(t)$ times.
    \end{itemize}
    That is, thinking of $f(t)$ as the number of times the agent uses the transition $t$, the agent enters and exits each connected component the same number of times, except that it exits the start location and enters the goal location once, and the agent uses the transitions of each gadget a consistent number of times.
  \end{definition}

  To prove containment in NP, our certificate that it is possible to reach the goal location is a controlled switching flow. Note that each gadget has a polynomial number of transitions, so $f$ has polynomially many entries.
  We need the following three lemmas:

  \begin{lemma}\label{lem:controlled switching flow convincing 1p}
    If there is a controlled switching flow, then it is possible to reach the goal location.
  \end{lemma}
  \begin{lemma}\label{lem:short controlled switching flow 1p}
    If it is possible to reach the goal location, then there is a polynomial-length controlled switching flow, i.e., one where $f(t)$ is at most exponential in the size of the system.
  \end{lemma}
  \begin{lemma}\label{lem:controlled switching flow checking 1p}
    There is a polynomial-time algorithm which determines whether a function $f$ is a controlled switching flow.
  \end{lemma}

  Together these imply that controlled switching flows can actually be used as certificates, and thus the one-player problem is in NP.

  \begin{proof}[Proof of Lemma~\ref{lem:controlled switching flow convincing 1p}]
    Let $f$ be a controlled switching flow. For each gadget $g$, pick a sequence of transitions of length $\ell_g=\sum\limits_{t\in g}f(t)$ in that copy which uses each transition $t$ exactly $f(t)$ times; this exists by the definition of a controlled switching flow. We play the one-player motion planning game in the system. Our strategy is based on the chosen sequences: whenever we arrive at a gadget, take the next transition in the sequence. If we find ourselves in a connected component with the input locations of multiple gadgets, we can enter any gadget $g$ which we have previously used fewer than $\ell_g$ times. We stop when we reach the connected component of the goal location, or when we have no moves obeying this strategy, meaning every gadget $g$ whose input location is currently reachable has already been used $\ell_g$ times.

    We claim this strategy must reach the goal location. If it does not, we must eventually get stuck with no moves (specifically, within $\sum\limits_t f(t)$ steps), and we will show this cannot happen because $f$ is a controlled switching flow. For the sake of contradiction, let $H$ be the connected component of the connection graph we are stuck in. To be stuck, we must have previously exited $H$ at least $\sum\limits_{t\in H_i} f(t)$ times. So we must have entered $H$ at least $\sum\limits_{t\in H_i} f(t)+1$ times (or one fewer, if the start location is in $H$). However, we have entered $H$ at most $\sum\limits_{t\in H_o} f(t)$ times, so $\sum\limits_{t\in H_o}f(t)\geq\sum\limits_{t\in H_i}f(t)+1$ or $\sum\limits_{t\in H_o}f(t)\geq\sum\limits_{t\in H_i}f(t)$ if the start location is in $H$, which violates the assumption that $f$ is a controlled switching flow.
  \end{proof}

  \begin{proof}[Proof of Lemma~\ref{lem:short controlled switching flow 1p}]
    For some path which reaches the goal location, let $f(t)$ be the number of times the path uses the traversal $t$. Then $f$ is clearly a controlled switching flow. The number of traversals in the shortest solution path is at most the number of configurations of the system of gadgets, which is at most $nk^n$ the system contains $n$ gadgets which have at most $k$ states. Thus using the shortest solution path, we have a controlled switching flow $f$ where $f(t)\leq nk^n$ and thus $f$ has polynomial length.
  \end{proof}

  \begin{proof}[Proof of Lemma~\ref{lem:controlled switching flow checking 1p}]
    The first condition for $f$ to be a controlled switching flow can be easily checked in polynomial time by computing the relevant sums.

    For the second condition, think of a gadget in the system as a directed multigraph with states as vertices and transitions as edges (labelled with their output locations).
		The second condition says that there is a walk through this graph starting at $s$ which uses each edge $t$ a specified number $f(t)$ of times. This is equivalent to an Euler tour in the (possible exponentially large) graph with $f(t)$ copies of the edge $t$. To verify that such a walk exists, we only need to check that the total in- and out-degrees match at each vertex (except possibly off by one at $s$ and one other vertex) and that the set of used transitions, i.e., those $t$ where $f(t)>0$, is connected. This can all be checked in polynomial time.
  \end{proof}

  This concludes the proof of Theorem~\ref{thm:1p in np}.
\end{proof}

\subsection{NP-hardness}\label{sec:1p np-hard}

In this section, we prove NP-hardness of one-player motion planning with each of the nontrivial single-input 2-state deterministic gadgets: the set switch and toggle switch. Our proofs can be easily adapted to prove NP-hardness of the corresponding problem for many input/output gadgets, but we leave open the problem of providing a characterization.

Our reduction is simpler than that in \cite{switchinggames}, and we show hardness for specific gadgets instead of general reachability switching games, which are equivalent to instance-specified gadgets.

\begin{theorem}\label{thm:1p np hard}
  One-player motion planning with either the toggle switch or the set switch is NP-hard.
\end{theorem}

\begin{proof}
  We provide essentially identical reductions from 3SAT to the two motion-planning problems. In the reduction, the player will never be able to traverse a gadget more than two times, so the difference between the toggle switch and the set switch is irrelevant. Each gadget will begin in the state which sends the agent to the `top' exit, and after a single traversal moves to the state which sends the agent to the `bottom' exit. We will describe the reduction in terms of the set-down switch, but it is equally applicable to the toggle switch.

  \begin{figure}
    \centering
    \begin{subfigure}{.7\linewidth}
      \centering
      \includegraphics[width=\linewidth]{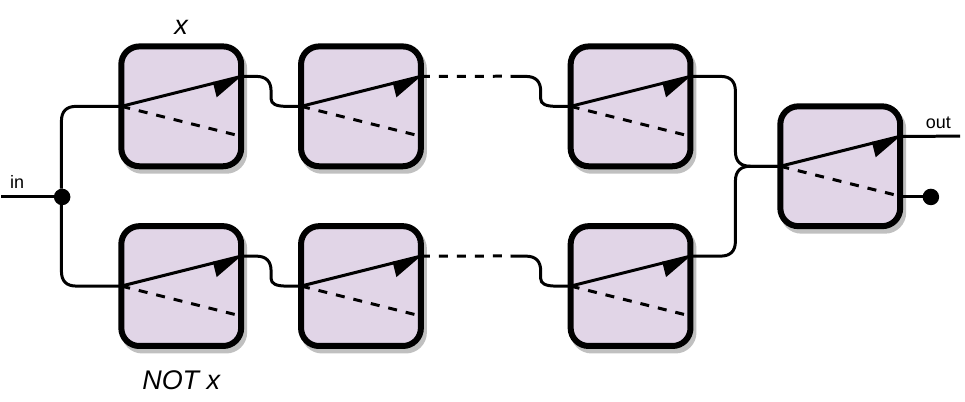}
      \caption{A variable. The player must pick one of the two branches to cross, and render the final set-down switch useless.}
    \end{subfigure}
    \hfill
    \begin{subfigure}{.25\linewidth}
      \centering
      \includegraphics[width=\linewidth]{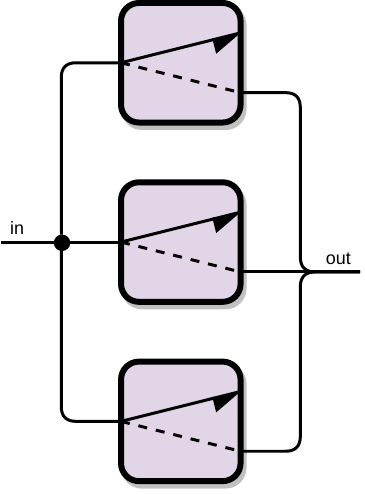}
      \caption{A clause. The player can only cross if at least one of the literals is set to true; otherwise they get stuck in a dead-end a variable.}
    \end{subfigure}

    \caption{Our reduction from 3SAT for one-player motion planning with the set-down switch.}
    \label{fig:3sat sd}
  \end{figure}

  For each variable in a 3SAT instance, there is a fork where the player may choose one of two paths. Each path passes through a series of set-down switches, exiting each from the top and setting them to the down state.
  The paths then merge and go through one more set-down switch, whose down exit is a dead-end.
  The number of gadgets in each branch depends on the number of instances of each literal in the formula.
  These variables are connected in series beginning at the start location, so the player is forced to walk through each variable, picking a side to use for each one.
  This corresponds to picking an assignment of each variable.
  After setting the variables in this way, the last set-down switch at the end of each variable is in the down state.

  For each clause, there is a 3-way fork, where the player must choose to go through one of the gadgets corresponding to a literal in the clause.
  If the chosen gadget was already traversed (during the variable-setting phase), the agent exits the bottom and can continue to the exit of the clause.
  If the chosen gadget was not already traversal, the agent exits the top, and finds itself in a variable.
  The player now has no choice but to walk down the variable path until the agent goes through the set-down switch at the end of the variable, which is in the down state, so the agent is now stuck in a dead-end.

  The clauses are connected in series, with the last variable leading to the first clause and the last clause leading to the goal location. In order to reach the goal location, the agent must pass through each variable and then each clause.
  In order to get through a clause without getting stuck, at least one gadget in the clause must have already been traversed; equivalently, at least one literal in the clause must be true under the assignment corresponding to the path taken during variable setting. Thus the agent can reach the goal location if and only if the formula has a satisfying assignment.

  For the set-down switch, once a gadget is in the down state it remains there forever, so it does not matter what order the clauses or gadgets within each variable path are in. However, for the toggle switch, when the agent walks through a clause the gadget it uses returns to the up state, which could lead to the agent later escaping that variable using the same gadget, instead of getting stuck in a dead-end.

  To prevent this, we order the gadgets on each variable path carefully.
  Specifically, we first choose an order for the clauses. For each literal $\ell$, we have the path corresponding to $\ell$ go through the set-down switches representing $\ell$ in the same order they appear in clauses.
  So if the agent moves from a clause to a variable through a gadget corresponding to $\ell$ (because $\ell$ was false), the agent cannot have previously interacted with any of the gadgets further along the line corresponding to $\ell$ during the clause-checking phase. In particular, those gadgets are all in the up state, and so the agent is in fact forced to go to the dead-end at the end of the variable.
\end{proof}

\begin{corollary}\label{cor:1p np hard}
  One-player motion planning with any nontrivial
  output-disjoint deterministic 2-state input/output gadget is NP-hard.
\end{corollary}

\begin{proof}
  It suffices to show that such a gadget can simulate
  either the toggle switch or the set switch,
  since Theorem~\ref{thm:1p np hard} says one-player motion planning
  with either of these gadgets is NP-hard.

  The only single-input output-disjoint deterministic 2-state input/output gadgets
  are the toggle switch and set switch, which simulate themselves.
  If the gadget is multi-input, by Lemma~\ref{lem:gadget basis}
  it simulates
  one of the eight basis gadgets shown in
  Figures~\ref{fig:bounded basis} and~\ref{fig:unbounded basis}.
  All but three of these contain a toggle switch or set switch.
  The switch/set-up line/set-down line contains the switch/set-up line,
  so we need only consider the switch/toggle and switch/set-up line.
  These can simulate, respectively, the toggle switch and the set-up switch,
  as shown in Figure~\ref{fig:sim single-input}.
\end{proof}

  \begin{figure}
    \centering
    \begin{subfigure}{.4\linewidth}
      \centering
      \includegraphics[scale=0.8]{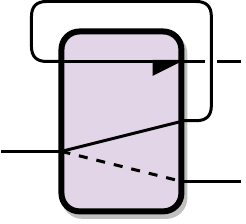}
      \caption{The switch/set-up line simulates the set-up switch.}
    \end{subfigure}
    \hfil
    \begin{subfigure}{.4\linewidth}
      \centering
      \includegraphics[scale=0.8]{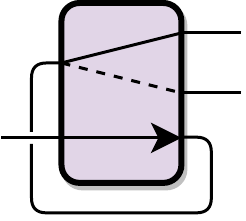}
      \caption{The switch/toggle line simulates the toggle switch.}
    \end{subfigure}
    \caption{Simulating nontrivial single-input gadgets
      from gadgets without a toggle- or set-switch.}
    \label{fig:sim single-input}
  \end{figure}

\section{Two Players}\label{sec:2p}

In this section, we consider a two-player game on systems of
input/output gadgets where the two players control a shared agent.
This game is analogous to the two-player reachability switching games of
\cite{switchinggames}, and we improve upon their results.
This game is different from two-player motion planning as defined in
\cite{DHL}, which has two agents each controlled by one player.

\begin{definition}
  For an input/output gadget $G$, \emph{two-player shared-agent motion planning with $G$} is played on a branchless system of $G$ and fan-out gadgets, with each gadget labeled Black or White, and has two players named Black and White. An agent begins at a designated start location. When the agent reaches a gadget, the player matching the gadget's label chooses a transition to take. White's goal is to reach a designated goal location, and Black's goal is to prevent this.

  The decision problem \emph{two-player shared-agent motion planning with $G$} is whether White has a strategy to force a victory.
\end{definition}

If $G$ is deterministic, the labels on $G$ do not matter: decisions are made only at fan-out gadgets.

Two-player reachability switching games are equivalent to two-player shared-agent motion planning with deterministic single-input input/output gadgets specified as part of the instance. It is shown in \cite{switchinggames} that this problem is in EXPTIME and PSPACE-hard.

In this section, we improve on this result in two ways.
First, we show that two-player shared-agent motion planning with any input/output gadgets (with any number of inputs) is in EXPTIME.
Second, we show that two-player shared-agent motion planning with just the toggle switch or the set switch is PSPACE-hard.
We give a reduction from Geography which is simpler than the reduction in \cite{switchinggames}.

We do not show EXPTIME-hardness for two-player shared-agent motion planning with any gadget.
We conjecture that two-player shared-agent motion planning is EXPTIME-hard with
any multi-input unbounded output-disjoint deterministic 2-state input/output gadget,
and perhaps even with any single-input such gadget (the only one being the toggle switch).

\begin{lemma}
  Two-player shared-agent motion planning with input/output gadgets is in EXPTIME.
\end{lemma}

\begin{proof}
  The two-player game can be simulated on an alternating Turing machine using polynomial space, where White's decisions are made by existential states and Black's decisions are made by universal states. Thus the problem is in APSPACE${}={}$EXPTIME.
\end{proof}

\begin{theorem}
  Two-player shared-agent motion planning with either the toggle switch or the set switch is PSPACE-hard.
\end{theorem}

\begin{proof}
  We provide essentially identical reductions from a version of Geography to the two problems. In the reduction, the agent will never be able to traverse a gadget more than two times, so the difference between the toggle switch and the set switch is irrelevant. As in the proof of Theorem~\ref{thm:1p np hard}, we will describe the reduction in terms of the set-down switch, but it is equivalent for the toggle switch.

  \emph{Vertex-Partizan Directed Vertex Geography} is a game played on a directed graph with specified start vertex, where each vertex is assigned to a player. Two players each move a marker along an edge whenever it is at one of their vertices, with the rule that the marker cannot visit the same vertex multiple times. A player loses if they have no moves. In Vertex-Partizan Max-Degree-3 Directed Vertex Geography, we assume every vertex has degree at most three, with at most two incoming edges and at most two outgoing edges. This problem is introduced and shown PSPACE-complete in \cite{bosboom2020edge}; vertex-partizan is a slight variation on bipartite Geography. We will refer to Vertex-Partizan Max-Degree-3 Directed Vertex Geography as simply \emph{Geography}.

  We construct an instance of two-player shared-agent motion planning with the set-down switch from an instance of Geography as follows. Each Geography vertex will be a single gadget, with tunnels in the connection graph corresponding to Geography edges. If a vertex has in-degree 1 and out-degree 2, we replace it with a fan-out gadget labeled with the player who is assigned that vertex. If a vertex has in-degree 2 and out-degree 1, we replace it with a set-down switch initially set to `up', with the `up' exit leading to the edge out of the vertex and the `down' exit leading to the goal location if the vertex is assigned to White and a dead-end otherwise. The agent starts at the location corresponding to the start vertex.

  Play on this system of gadgets proceeds as follows. When the agent reaches a fan-out gadget, the player assigned the corresponding vertex chooses which output to take. When the agent reaches a set-down switch for the first time, it continues along the outgoing edge to another vertex. When it reaches a set-down switch for the second time, the game ends and the player who is assigned the corresponding vertex wins. This is the same as the game of Geography: a player loses if the agent moves from one of their vertices to a set-down switch it visited before, which is equivalent to players not being allowed to move the marker to an already-visited vertex.
\end{proof}

\section{Applications}
\label{Applications}

In this section, we use the results in this paper to prove PSPACE-completeness
of the mechanics in several video games: one-train colorless Trainyard, [the Sequence], trains in Factorio, and transport belts in Factorio.\footnote{Factorio in general is already known to be PSPACE-complete, as players have explicitly built computers using the circuit network; for instance, see \url{https://forums.factorio.com/viewtopic.php?t=42708} or \url{https://redd.it/6imjhv}. We consider the restricted problems with only train-related objects and only transport belt-related objects.}
For each of these problems, the decision problem is the long-term behavior of a deterministic system, e.g., whether a train ever reaches a specific location. Another interesting decision problem, which we do not consider here, is the design problem: given some set of constraints, is it possible to build a configuration with a desired behavior? This is perhaps more natural because, for example, it captures the question of deciding whether a level in Trainyard is solvable.

\subsection{Trainyard}

The study of the complexity of Trainyard began with \cite{almanza2018trainyard}, which showed that finding a solution to a Trainyard level is NP-hard. Later, \cite{trainyard2} showed that checking a solution to a Trainyard level is PSPACE-complete---verifying solutions may be harder than finding them. We improve on this result by showing that checking a solution to a Trainyard level is PSPACE-hard even with only one train, and with no color changes.

Trainyard is a puzzle game in which the goal is to build a system of rails so that trains of the correct colors reach certain stations.
We consider one-train colorless Trainyard, where solutions consist of only rails, crossings, and \emph{switches}%
\footnote{This use of the word `switch' is unrelated to the component of input-output gadgets.}
There is a single train which moves forward along the rails; it succeeds if it reaches a designated location, and crashes and fails if the track it is on ends.
Rails can be traversed in both directions.

The only nontrivial behavior comes from switches, which have two states. A switch changes state every time the train moves through it. It has three locations: two of them always route the train to the third, and the third routes the train to one of the first two depending on the state. We can model this as a toggle line/toggle line/toggle switch with some locations identified; we call this the \emph{Trainyard gadget}, which is shown in Figure~\ref{fig:trainyard gadget}. Since tracks can bend and cross each other, the planarity of a system of Trainyard gadgets does not matter. Now one-train colorless Trainyard is equivalent to zero-player motion planning with the Trainyard gadget---except that the Trainyard gadget is not input/output, so we have not defined zero-player motion planning with it.

\begin{figure}
  \centering
  \includegraphics[scale=.8]{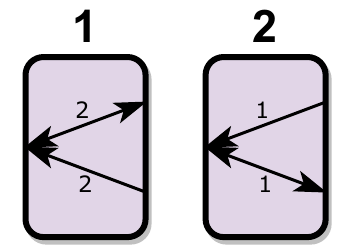}
  \caption{The Trainyard gadget.}
  \label{fig:trainyard gadget}
\end{figure}

\begin{definition}
  \emph{Zero-player motion planning with the Trainyard gadget} takes place in a system of Trainyard gadgets where the connection graph is a partial matching. That is, each location is either paired with one other location or a \emph{dead-end}.

  An agent moves through the system similarly to with input/output gadgets. When it enters a Trainyard gadget, it takes the unique available transition. When it exits a Trainyard gadget, it moves to the unique paired location, or stops and \emph{crashes} if it is at a dead-end.
\end{definition}

\begin{theorem}
  Zero-player motion planning with the Trainyard gadget, or equivalently checking a solution to one-train colorless Trainyard, is PSPACE-hard.
\end{theorem}

\begin{proof}
  We will reduce from zero-player motion planning with the toggle switch/toggle line.
  Formally, we actually reduce from a restricted version of this problem, where we are promised that the agent does not enter an infinite loop---it either reaches the goal location or a dead-end. The system constructed by Theorem~\ref{thm:SUDPSPACE} satisfies this property, and the simulations in Section~\ref{sec:simulating sud} preserve it, so the restricted promise problem is still PSPACE-hard.

  We cannot quite directly simulate a toggle switch/toggle line, for a few reasons:
  \begin{itemize}
    \item The Trainyard gadget, and thus any gadget simulated by it, can be entered at any location, not just input locations. To account for this, we will denote some vertices in the simulation as input and output, and the arrangement of gadgets will ensure that the agent always enters simulated gadgets at input-denoted locations and exits and output-denoted locations. In particular, paths emerging from output-denoted locations always lead to input-denoted locations.
    \item Zero-player motion planning with the Trainyard gadget does not include fan-ins. However, we can easily simulate fan-in in the above sense by denoting two locations as input and one as output on the Trainyard gadget---the Trainyard gadget is a fan-in provided the agent never enters at the left location.
    \item Even with the above caveats, we have not been able to simulate the toggle switch/toggle line (or any unbounded output-disjoint deterministic 2-state input/output gadget with multiple nontrivial inputs) with the Trainyard gadget.
      Instead, we simulate a toggle switch/toggle line for \emph{exponentially long}.
      Formally, we describe a network of Trainyard gadgets for each natural number $k$ such that the $k$th network has the same behavior as the toggle switch/toggle line for at least $2^k$ transitions, and can be constructed in time polynomial in $k$.
      Consider a system of $n$ toggle switch/toggle lines in which the agent never enters an infinite loop
      (such as the one used to prove PSPACE-hardness).
      The system has at most $2^n$ configurations and $5n$ locations for the agent; thus after at most $5n2^n$ transitions, the agent reaches either the goal location or a dead-end.
      If we pick a polynomial-size $k$ such that $2^k>5n2^n$ (e.g., $k=2n+3$ suffices), then the network of Trainyard gadgets we obtain by replacing each toggle switch/toggle line with the $k$th simulation has the same behavior long enough for the agent to either reach the goal location or crash at a dead-end.
      Hence these exponentially long simulations suffice for PSPACE-hardness.
  \end{itemize}

  Thus it suffices to find an exponentially long simulation of the toggle switch/toggle line. Before describing this simulation, we present an exponentially long simulation of an intermediate gadget called the \emph{reverse branch}, shown in Figure~\ref{fig:rev branch}. This gadget has one state and three locations.
  We assume the top right location is only used as an entrance,
  and the bottom right location is only used as an exit.

  \begin{figure}[htbp]
    \centering
    \includegraphics[scale=.8]{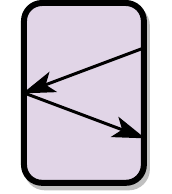}
    \caption{The reverse branch gadget.}
    \label{fig:rev branch}
  \end{figure}

  Figure~\ref{fig:trainyard to rb} shows our exponentially long simulation of a reverse branch. The $k$ gadgets in the bottom row serve as fan-ins, since we assume the agent never enters at the bottom right. Consider the states of the top row of $k+1$ gadgets as describing a number in binary: up (state 1) is $0$, down (state 2) is $1$, and the bits are read right to left. When the agent enters at the left, it increments this number (mod $2^{k+1}$) and exits at the bottom right, unless the states are all up so the number is $0$, in which case it exits the top right. When the agent enters at the top right, it flips the state of every gadget in the top row and exits at the left; this changes the number by $x\mapsto-x-1$. In particular, the distance from $0$ changes by at most $1$ with each transition. By starting at $2^k$ as in Figure~\ref{fig:trainyard to rb}, it takes at least $2^k$ transitions to reach $0$, so the simulation is correct for $2^k$ transitions.

  \begin{figure}
    \centering
    \includegraphics[scale=.8]{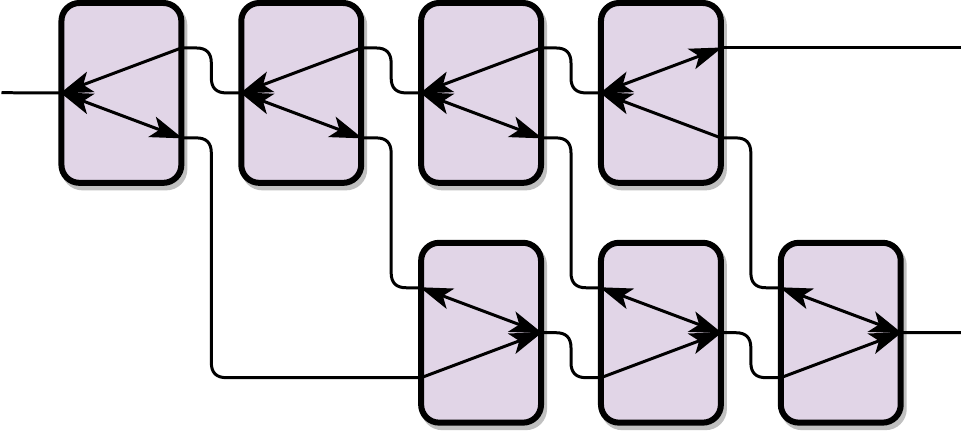}
    \caption{An exponentially long simulation of a reverse branch using Trainyard gadgets.}
    \label{fig:trainyard to rb}
  \end{figure}

  Now we simulate a toggle switch/toggle line using a Trainyard gadget and two reverse branches, as shown in Figure~\ref{fig:trainyard tw/t}. When the agent enters at $a$, it exits at $b$, flipping the state of the Trainyard gadget (in the middle); this is the toggle line. When the agent enters at $c$, it exits at $d$ or $e$ depending on the state of the Trainyard gadget, and flips the state; this is the toggle switch. Each transition through the simulated gadget makes at most one transition through each reverse branch, so if the reverse branches are correct for $2^k$ transitions, so is the toggle switch/toggle line.
\end{proof}

  \begin{figure}
    \centering
    \includegraphics[scale=.8]{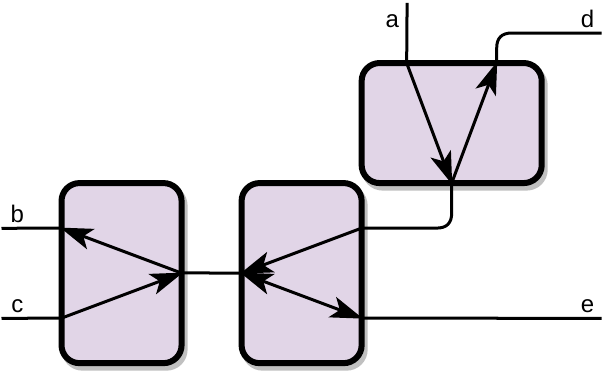}
    \caption{A simulation of a toggle switch/toggle line using a Trainyard gadget and reverse branches.}
    \label{fig:trainyard tw/t}
  \end{figure}

\subsection{[the Sequence]}

[the Sequence] is a puzzle game in which the player attempts to place \emph{modules} to move \emph{binary units} from a \emph{source} to a \emph{target}. There are seven different modules which have different effects; for instance, the \emph{pusher} moves anything immediately in front of it one square away.%
\footnote{Module names are not given in the game, so we use our own names.}
In this section, we prove that determining the correctness of a solution to a [the Sequence] puzzle is PSPACE-complete.

First we describe the mechanics of [the Sequence] that are necessary for our proof. The game takes place on a bounded square grid, containing the source and target, some fixed blocks, and some modules (which the player has placed, and which have an orientation). A deterministic simulation occurs in a series of rounds. Each round begins with the source creating a binary unit if it does not already have one. Then each module acts in a specified order; their actions are detailed below. Only modules and binary units can be moved. A binary unit disappears when it reaches the target. If objects (binary units, modules, or walls) ever collide, the simulation stops. The solution is correct if it moves an arbitrary number of binary units from the source to the target.%
\footnote{The game checks that four binary units are successfully moved, but an unlimited number is more natural.}

The modules used in our proof are the following, shown in Figure~\ref{fig:seq modules}:
\begin{itemize}
  \item The \emph{mover} moves one square forward, bringing any module or binary unit immediately to its left with it.%
  \footnote{These modules can also be reflected, but we do not use that.}
  \item The \emph{turner} rotates any module immediately in front of it $90^\circ$ counterclockwise.
  \item The \emph{puller} moves any module or binary unit two squares in front of it to only one square in front of it.
\end{itemize}

\begin{figure}
  \centering
  \begin{subfigure}{.24\linewidth}
    \centering
    \includegraphics[width=.3\linewidth]{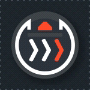}
    \caption{The mover module.}
  \end{subfigure}
  \begin{subfigure}{.24\linewidth}
    \centering
    \includegraphics[width=.3\linewidth]{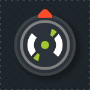}
    \caption{The turner module.}
  \end{subfigure}
  \begin{subfigure}{.24\linewidth}
    \centering
    \includegraphics[width=.3\linewidth]{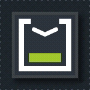}
    \caption{The puller module.}
  \end{subfigure}
  \begin{subfigure}{.24\linewidth}
    \centering
    \includegraphics[width=.3\linewidth]{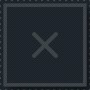}
    \caption{A fixed block.}
  \end{subfigure}

  \caption{[the Sequence] modules used in our proof, and the fixed block.}
  \label{fig:seq modules}
\end{figure}

\begin{theorem}
  Determining correctness of a solution to a [the Sequence] puzzle is PSPACE-complete.
\end{theorem}

We prove hardness using a reduction from zero-player motion planning with the switch/set-up line/set-down line. Our proof is robust to the definition of correctness, in the following sense: if the agent reaches the goal location, the solution moves arbitrarily many binary units to the target. If the agent does not reach the goal location, the solution runs forever without moving any binary units. A simple modification to the reduction makes the simulation crash if the agent does not reach the goal (though this requires the property of the proof of PSPACE-hardness of zero-player motion planning that the agent reaches a specific location exactly when it does not reach the goal).

\begin{proof}
  The game is a deterministic simulation with a polynomial amount of state (each square has at most one module or binary unit, which takes a constant amount of memory), so the simulation can be carried out in polynomial space. Determining whether arbitrarily many binary units will be delivered to the target is harder, but can still be done in PSPACE by detecting a cycle in the configuration, perhaps using a tortoise-and-hare algorithm.

  For PSPACE-hardness, we give a reduction from zero-player motion planning with the switch/\allowbreak set-up line/set-down line. The agent is represented by a single mover. Turners rotate the mover to control its path. Fan-in is accomplished using a puller to merge to adjacent parallel paths, as shown in Figure~\ref{fig:seq fanin}. Paths can easily cross each other since they only require modules at corners and fan-ins.

  \begin{figure}
    \centering
    \includegraphics[width=.5\linewidth]{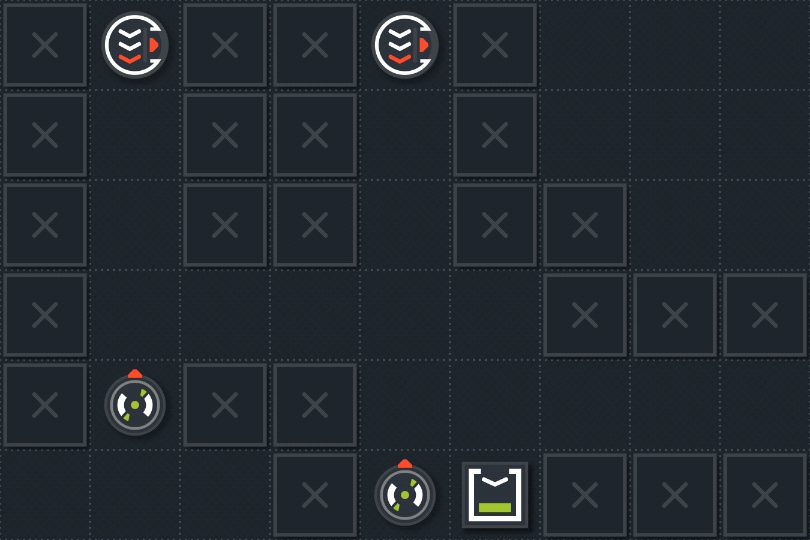}
    \caption{Bending paths and a fan-in for [the Sequence]. The movers represent where the mover might enter; they do not simultaneously exist. If the mover enters in either location, it is turned and possibly pulled, and then exits the right. The fixed blocks are only to help visualize the paths taken.}
    \label{fig:seq fanin}
  \end{figure}

  The switch/set-up line/set-down line, shown in Figure~\ref{fig:seq w/su/sd}, is built using three pullers. When the agent enters the set-up or set-down line, the mover moves the central puller to a particular side. When the agent enters the switch, the mover is pulled if the puller is on the appropriate side; the path it exits depends on the state of the gadget.

  \begin{figure}
    \centering
    \begin{subfigure}{.45\linewidth}
      \centering
      \includegraphics[width=.85\linewidth]{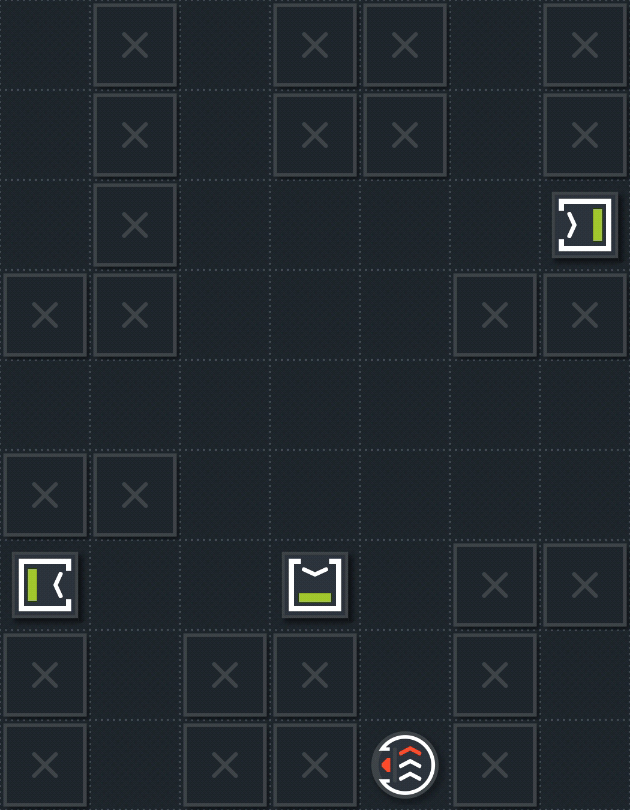}
      \caption{The gadget in the down state, with a mover entering to set it to the up state.}
    \end{subfigure}
    \hfill
    \begin{subfigure}{.45\linewidth}
      \centering
      \includegraphics[width=.85\linewidth]{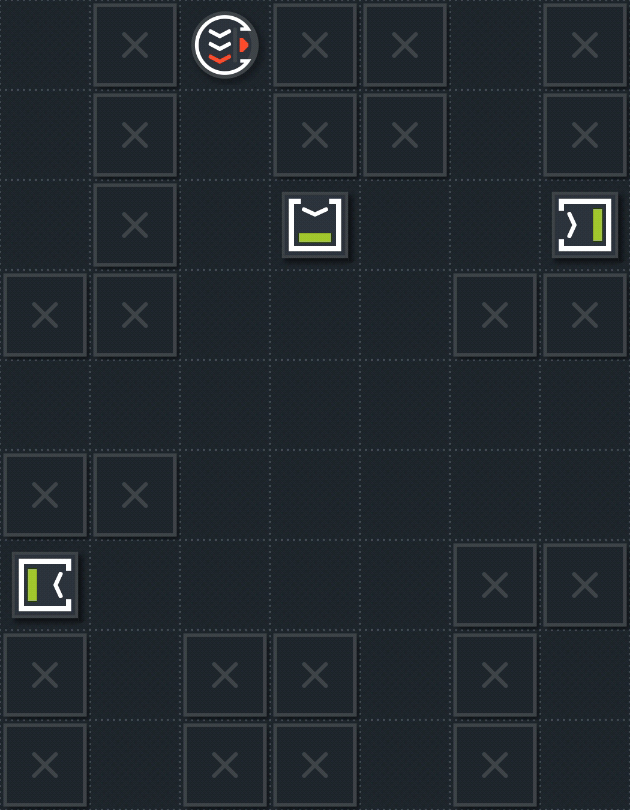}
      \caption{The gadget in the up state, with a mover entering to set it to the down state.}
    \end{subfigure}
    \begin{subfigure}{.45\linewidth}
      \centering
      \includegraphics[width=.85\linewidth]{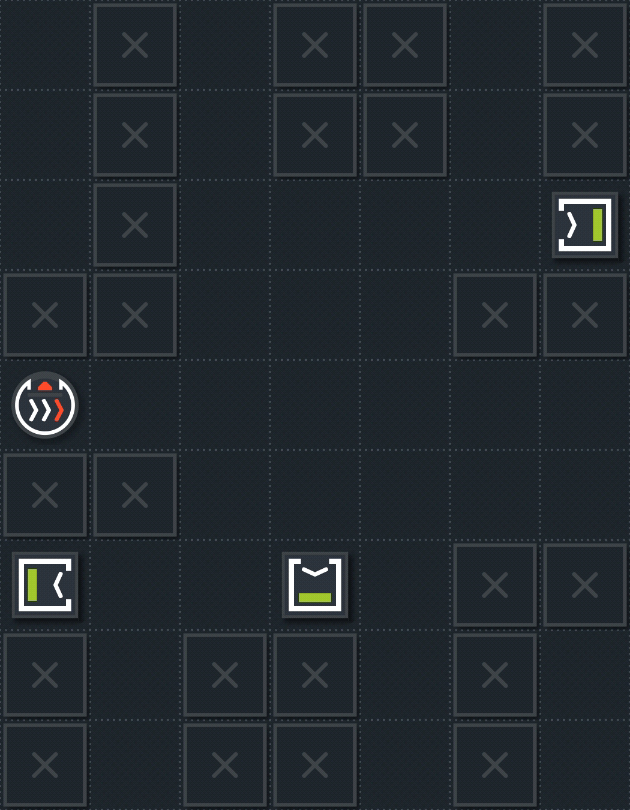}
      \caption{The gadget in the down state, with a mover entering the switch.}
    \end{subfigure}
    \begin{subfigure}{.45\linewidth}
      \centering
      \includegraphics[width=.85\linewidth]{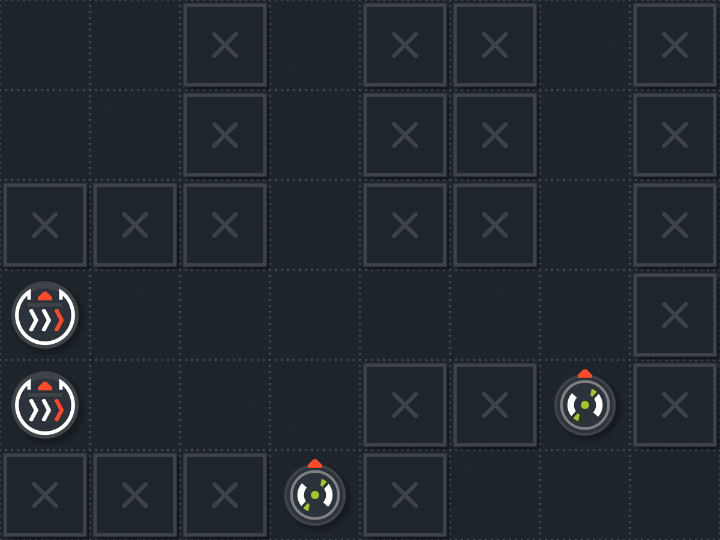}
      \caption{Separating the adjacent paths, analogous to Figure~\ref{fig:seq fanin}.}
    \end{subfigure}
    \caption{A switch/set-up line/set-down line for [the Sequence]. The puller in the middle encodes the state. When the mover enters the bottom right (a) or top left (b), it moves the middle puller to the top or bottom and then gets pulled away, setting the state. When it enters the left, which row it exits on the right depends on whether the middle puller was at the bottom to pull the mover down. Finally, we separate these two paths with appropriate turners (d).}
    \label{fig:seq w/su/sd}
  \end{figure}

  If the agent reaches the goal location, the mover reaches a cycle which has it deliver binary units to the target, shown in Figure~\ref{fig:seq win}. Otherwise it gets stuck in the maze of modules forever, and never moves any binary units.
\end{proof}

  \begin{figure}
    \centering
    \includegraphics[width=.5\linewidth]{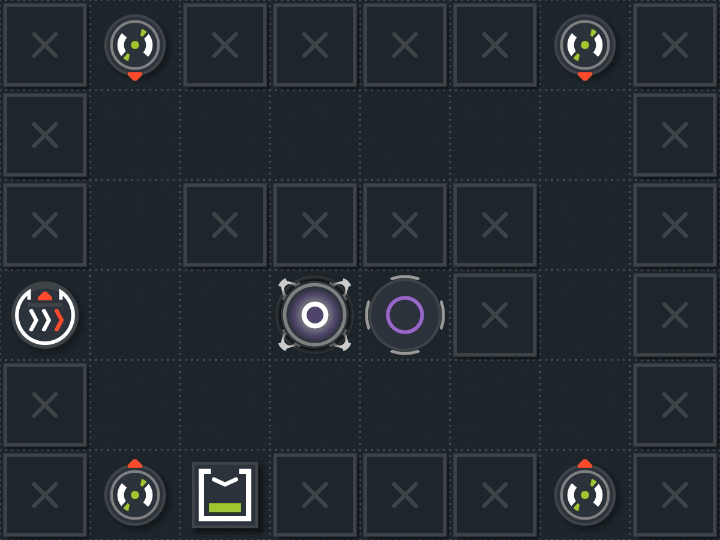}
    \caption{A cycle for the mover to deliver binary units. The two objects in the middle are the source and target. If the mover enters at the left as shown, it is pulled into the loop similarly to the fan-in. The turners then keep it moving in a rectangle, and it moves a binary unit from the source to the target on each loop.}
    \label{fig:seq win}
  \end{figure}

\subsection{Factorio Trains}
\label{sec:FactorioTrains}
Our first application for the factory-building video game Factorio is showing that trains in Factorio are PSPACE-complete. The decision problem we consider is whether a particular \emph{target train} ever reaches its target station, in a world with only a train system and no player interaction. Other work on the computational power of Factorio Train systems includes the simulation of cellular automata Rule 110 on a bounded tape \cite{110trains}. The logical infrastructure used to implement Rule 110 is significantly more sophisticated and is likely sufficient to show PSPACE-completeness given proper analysis. We provide our own construction and prove PSPACE-completeness by reducing from zero-player motion planning with the switch/set-up line/set-down line.

Next we describe the train components of Factorio,
illustrated in Figure~\ref{fig:train components}.
Trains in Factorio are constrained to \emph{rails}, which can bend, fork, and cross each other. \emph{Stations} are locations which trains will try to reach. Each train is provided with a \emph{schedule}, which is a list of stations; the train will move to each station in the list in cyclic order. If there are multiple stations with the same name, the train will find the cheapest path to any of them, where the cost of a path depends on its length and also on properties including the number of other trains blocking the path and the amount of time the train has been waiting so far. Trains are prevented from crashing into each other using \emph{rail signals} and \emph{chain signals}. These partition the rail system into \emph{blocks}, and (roughly) trains will not enter a block that is already occupied by another train.

\begin{figure}
  \centering
  \includegraphics[width=.8\linewidth]{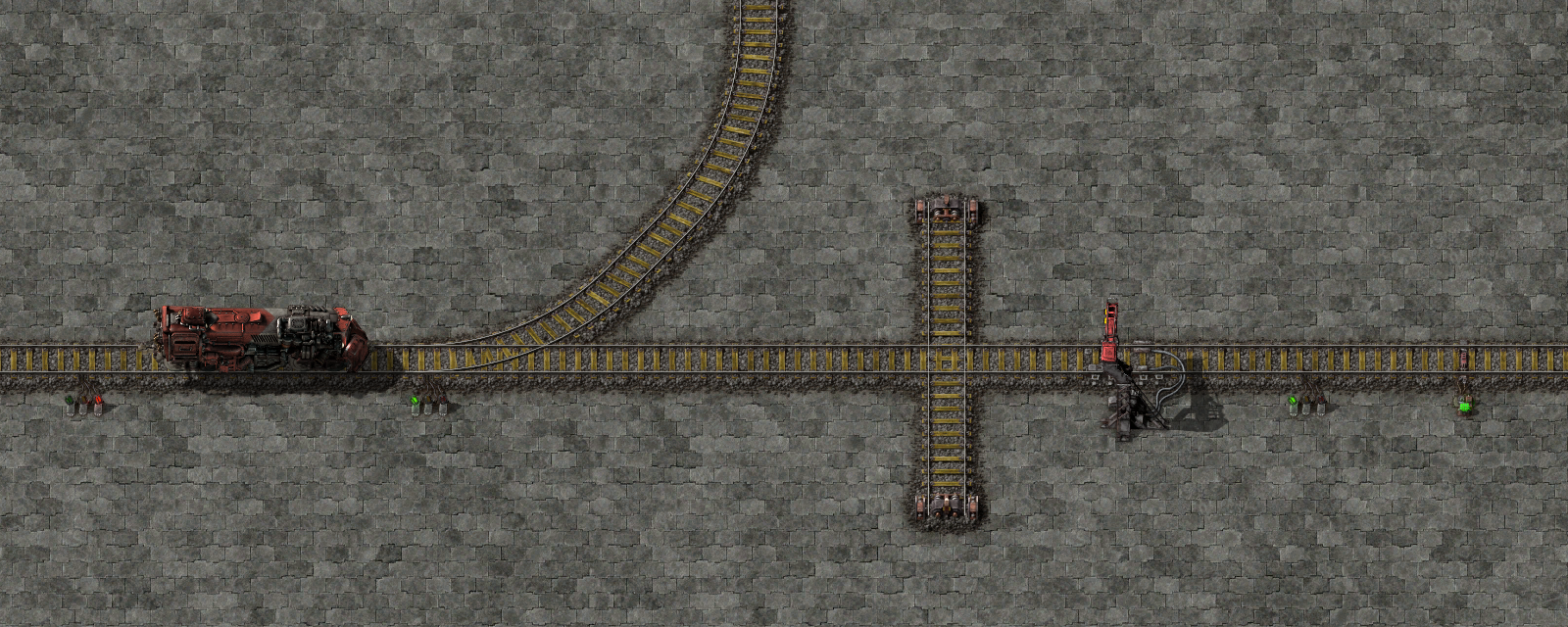}
  \caption{A demonstration of the train-related objects in Factorio. From left to right, we have a rail signal, locomotive, rail signal, rail fork, rail crossing, train station, rail signal, and chain signal. The leftmost rail signal is red, indicating the presence of a train in the block in front of it.}
  \label{fig:train components}
\end{figure}

Here are some caveats applying to our hardness proof:
\begin{itemize}
  \item We assume that the only objects in the world are rails, locomotives (we make no use of cargo wagons), train stations, rail signals, and chain signals. In particular, there are no players, construction robots, circuit networks, or biters.
  \item We ignore fuel requirements of trains, assuming they have unlimited fuel. Without this assumption, and without allowing some mechanism to provide fuel, the problem would be in NP (since each train would move a bounded distance before running out of fuel).
  \item We do not use the complexity in train wait conditions: the only wait condition used is \texttt{0~seconds}. In fact, every train's schedule consists of just two stations \texttt{A} and \texttt{B}, which the train will alternate between.
  \item We do not know all of the details of the behavior of trains, but under the plausible assumptions that a single game tick is simulated in polynomial time and the amount of memory associated with each train (and other train-related component) is polynomially bounded, the decision problem is clearly in PSPACE.
\end{itemize}

\begin{theorem}
  In a Factorio world with only rails, locomotives, train stations, rail signals, and chain signals, and where each train's schedule alternates between the same two stations \texttt{A} and \texttt{B} with the trivial wait condition, it is PSPACE-hard to determine whether a specified target train ever reaches its next station.
\end{theorem}

\begin{proof}
  We show PSPACE-hardness through a reduction from zero-player motion planning with the switch/set-up line/set-down line. The rail network is mostly full of trains, one in each block, and the motion-planning agent is represented by a \emph{gap} (not by a train), a block which does not contain a train and thus which a train can move into. A simple `wire' is constructed by a line of blocks all occupied by trains, as shown in Figure~\ref{fig:train wire}. When the gap reaches the front of the line, each train in turn is able to move forward one block, moving the gap to the end of the line. The movement of the gap is in the opposite direction of the movement of the trains.

  \begin{figure}
    \centering
    \includegraphics[width=.8\linewidth]{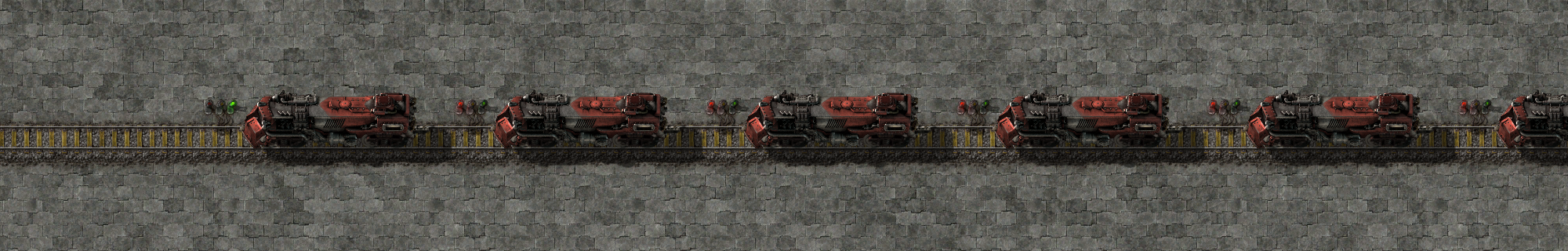}
    \caption{A wire of Factorio trains. The leftmost rail signal is green, so the leftmost train is free to move left. Then the train behind it can also move left, and so on, moving the gap right.}
    \label{fig:train wire}
  \end{figure}

  Fan-ins are achieved using a fork in a track, as shown in Figure~\ref{fig:train fanin}. When the gap arrives at either branch of the fork, the train just entering the fork will move forward to fill the gap, since trains prefer paths with fewer other trains in the way. To ensure the path-finding works as expected, we place stations before and after the fork which make the paths the train at the fork needs to find short.

  \begin{figure}
    \centering
    \includegraphics[width=.4\linewidth]{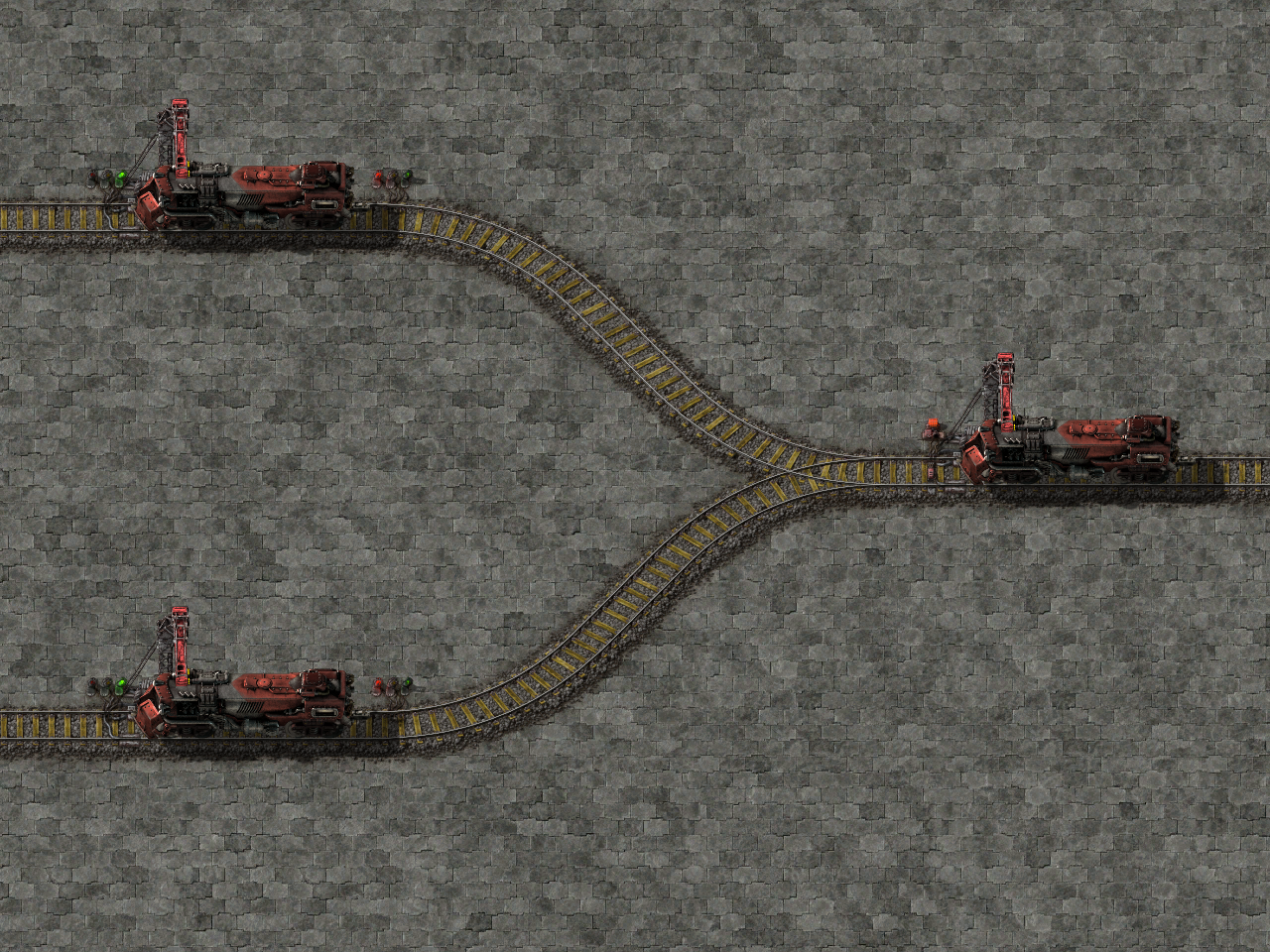}
    \caption{A fan-in for Factorio trains. Both stations on the left are named \texttt{B} and the station on the right is named \texttt{A}. When either train on the left leaves, the train on the right will fill its spot: it takes the cheapest path to \texttt{B}, and paths blocked by trains are considered more expensive. The chain signal immediately before the fork prevents the train from choosing a branch before one of them is empty.}
    \label{fig:train fanin}
  \end{figure}

  Since the network of gadgets in zero-player motion planning may be nonplanar, we need a crossover. This is easy to build using two crossing rail lines with the appropriate configuration of rail and chain signals, shown in Figure~\ref{fig:train crossover}.

  \begin{figure}
    \centering
    \includegraphics[width=.3\linewidth]{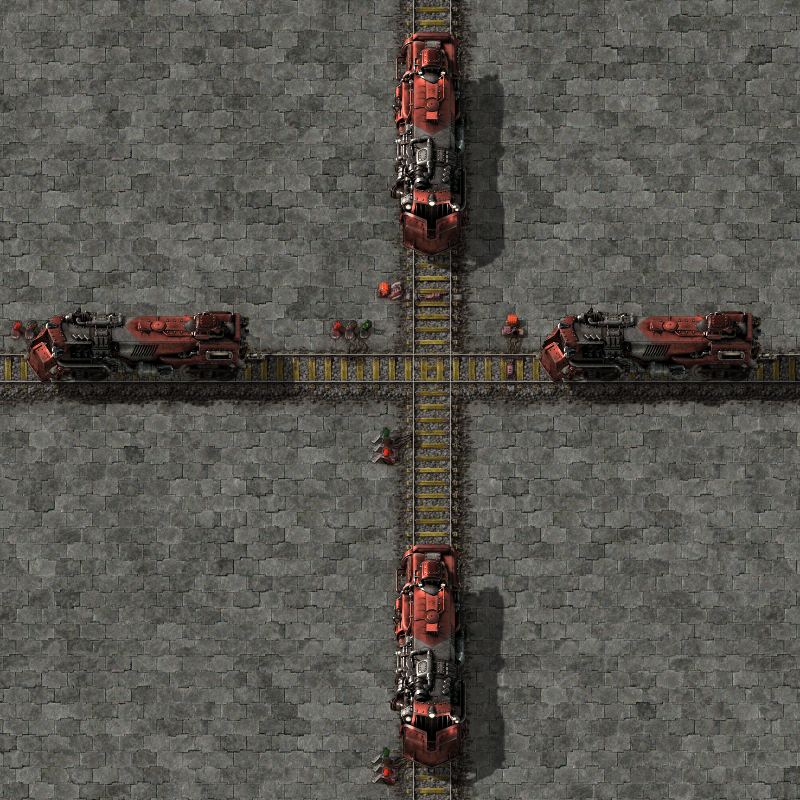}
    \caption{A crossover for Factorio trains. The two crossing wires can move independently. Chain signals prevent trains from blocking the intersection until the train in front moves out of the way.}
    \label{fig:train crossover}
  \end{figure}

  The initial location of the agent is represented by an empty block---in fact, the only empty (nonchain signaled)%
  \footnote{There is also an empty block at each intersection. However, these blocks have chain signals at their entrances, so trains will never stop in them, and the fact that they are empty does not allow any trains to move.}
  block in the network. We place the target train on a short rail line blocked by the train in the block representing the goal location, as shown in Figure~\ref{fig:train win}: the target train will move forward and reach its target station if and only if the train blocking it moves, which happens exactly when the agent reaches the goal location.

  \begin{figure}
    \centering
    \includegraphics[width=.4\linewidth]{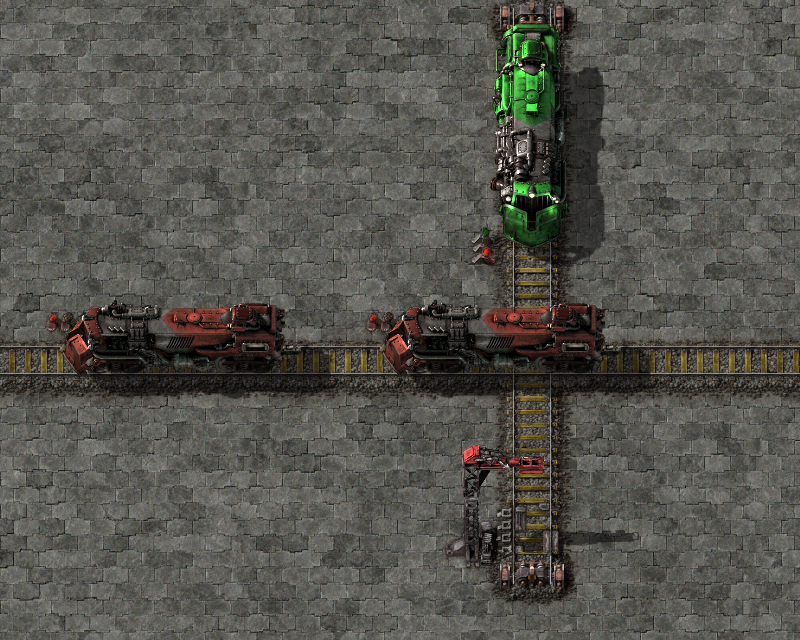}
    \caption{The win gadget for Factorio trains. Once the gap arrives on the left, the green train will be able to move forward and reach the station.}
    \label{fig:train win}
  \end{figure}

  Finally, we need to build the switch/set-up line/set-down line using trains. This gadget is shown in Figure~\ref{fig:train w/su/sd}. We have a train trapped in a loop in the gadget, which encodes the state. When the gap traverses the set-up or set-down line, the trapped train is temporarily not blocked and moves forward. When the gap enters the switch, the output it takes depends on what the trapped train is currently blocking.
\end{proof}

  \begin{figure}
    \centering
    \begin{subfigure}{.45\linewidth}
      \includegraphics[width=\linewidth]{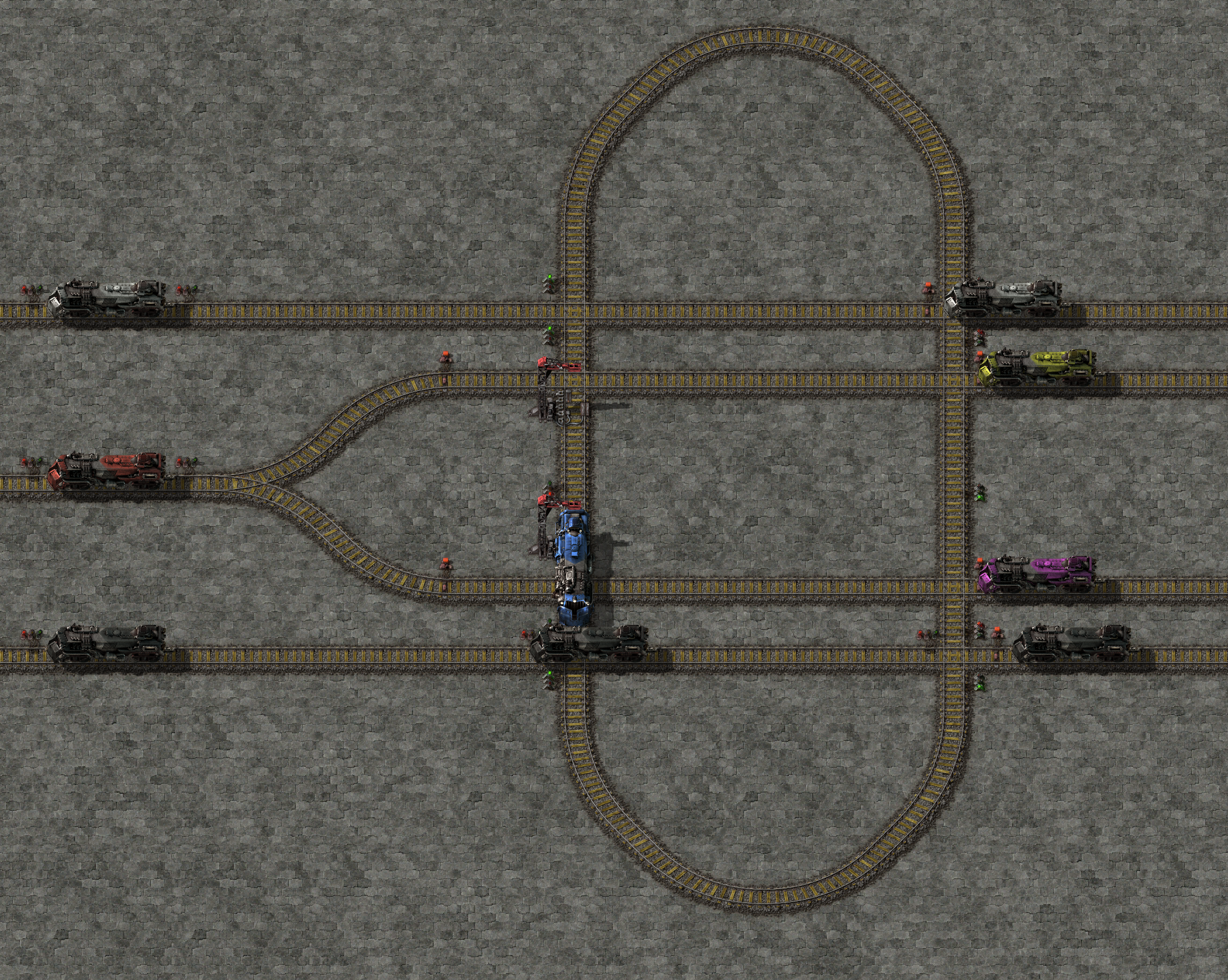}
      \caption{The gadget in the up state.}
    \end{subfigure}
    \begin{subfigure}{.45\linewidth}
      \includegraphics[width=\linewidth]{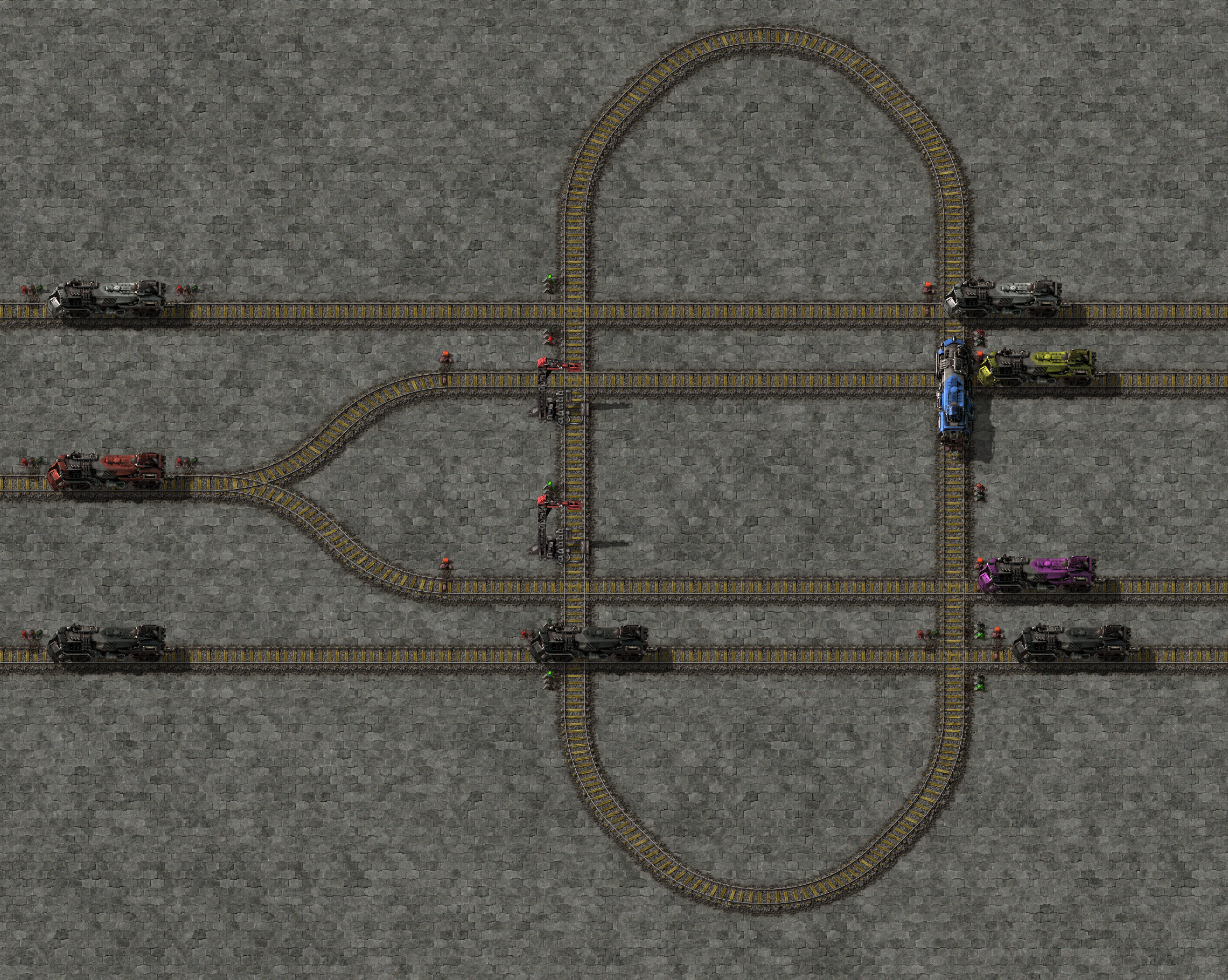}
      \caption{The gadget in the down state.}
    \end{subfigure}
    \caption{A switch/set-up line/set-down line made of Factorio trains. The blue train is trapped in the loop and encodes the state of the gadget; it alternates between the two train stations in the loop. From the down state, if the gap enters the top track of white trains, the blue train is briefly not blocked and moves to where it is blocked by a train in the bottom track. Similarly when the bottom track advances the blue train moves to be blocked by a train in the top track. When the gap enters the middle track and the red train leaves, whichever of the yellow or purple trains is not blocked by the blue train moves forward.}
    \label{fig:train w/su/sd}
  \end{figure}

\subsection{Factorio Transport Belts}

In this section, we show that determining whether an item ever reaches a goal location in a Factorio world with only transport belts, underground belts, and splitters is PSPACE-complete. This result holds even with a small bounded number of mobile items: in Factorio 0.15 and earlier, we use a single mobile item (and a polynomial number of immobile items for a technical reason), and in Factorio 0.16 and later, we use two mobile items (and no immobile items).

\emph{Transport belts} move any items on them in the direction the belt is facing. Items move smoothly between transport belts, but will stop if there is not another transport belt (or similar object) in front---transport belts  will not dump items onto the ground. \emph{Underground belts} can be used to have belts cross; two matching underground belts with at most four tiles between them will transfer items ``underground''.%
\footnote{The distance is longer for fast and express underground belts, but we do not need them.}
\emph{Splitters} can have up to two transport belts feeding in and two transport belts feeding out. Splitters alternate which output they send items to, regardless of the input they came from, except that if one output is blocked by items, all items will go to the other output. The details of the alternation changed slightly in Factorio version 0.16; we will explain and investigate the complexity of both versions of the mechanic.

Transport belts have two \emph{lanes}, one on each side of the belt, which move items independently, are preserved going around corners and through splitters. A transport belt facing into the side of another transport belt will deliver items to the nearer lane; this is called \emph{sideloading}. When sideloading onto an underground belt, only one lane of the incoming belt is able to move; the other lane is blocked.

\begin{figure}
  \centering
  \includegraphics[width=.5\linewidth]{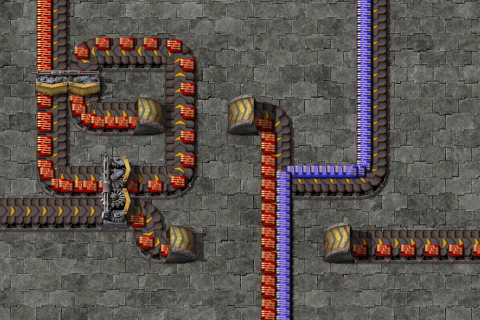}
  \caption{An example transport belt layout, demonstrating transport belts, underground belts, splitters, lanes, and sideloading.}
  \label{fig:belt example}
\end{figure}

As with trains, containment in PSPACE is trivial assuming each game tick is simulated in polynomial time.

\begin{theorem}
  In a Factorio world with only transport belts, underground belts, splitters and items, it is PSPACE-hard to determine whether an item ever reaches a goal location. In 0.15 and earlier, this remains PSPACE-hard when only one item can move and a polynomial number of items are stuck. In 0.16 and later, this remains PSPACE-hard when there are only two items.
\end{theorem}

\begin{proof}

  For both versions of splitter behavior, we show PSPACE-hardness through a reduction from zero-player motion planning with the toggle switch/toggle switch. Wires are simply chains of transport belts, fan-in is accomplished by sideloading, and crossovers can be built using underground belts.

  The toggle switch/toggle switch is different for the two versions, and depends on the details of splitter behavior.

  \paragraph{Factorio 0.15 and earlier.}

  Prior to 0.16, splitters alternate both lanes together and each item type separately. For each item type, all items of that type entering the splitter will alternate which output belt they take regardless of the lane they are on. The lane an item is on is preserved.%
  \footnote{Since each item type alternates independently, the splitter requires one bit of state for each item type. One can take advantage of this complexity for tasks including sorting items; we will not use it because there will be only one mobile item.}
  We view a splitter as having two lanes as inputs, and four outputs: two lanes on each of two belts. The splitter then behaves as a toggle switch/toggle switch---each lane is a toggle switch, and they share a state.%
  \footnote{Really it is a separate toggle switch/toggle switch for each item type, but we will have only one mobile item.}

  We need to make a toggle switch/toggle switch which uses transport belts instead of lanes as inputs and outputs. This can be accomplished using sideloading onto the correct lane for the inputs and sideloading onto an underground belt for the outputs, shown in Figure~\ref{fig:belt 0.15 tw/tw}.

  \begin{figure}
    \centering
    \includegraphics[height=7cm]{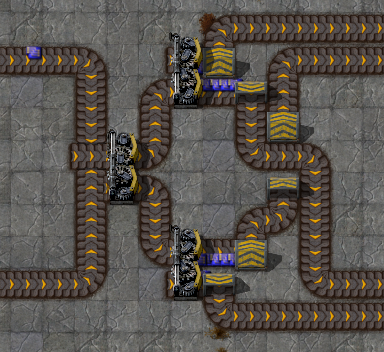}
    \hfil
    \includegraphics[height=7cm]{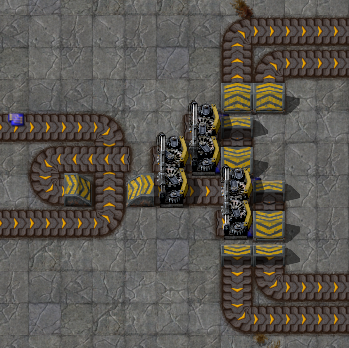}
    \caption{Two layouts of a toggle switch/toggle switch for transport belts in Factorio 0.15 and earlier. An incoming item is put on a lane depending on the input it enters. It passes through the leftmost splitter, which encodes the state of the gadget. The other two splitters help separate lanes: one lane of each output belt is blocked by items (by sideloading onto underground belts), so the output belt is determined by the lane of an item that enters that splitter. This gadget requires a constant number of immobile items: the layout on the left uses 16, and the layout on the right uses only 8 (but is harder to parse).}
    \label{fig:belt 0.15 tw/tw}
  \end{figure}

  The initial state of each toggle switch/toggle switch is encoded by the state of a splitter. We place a single item at the start location, and it simulates the agent in zero-player motion planning, reaching the goal location if and only if the agent does.

  \paragraph{Factorio 0.16 and later.}

  In 0.16, splitters were changed to alternate each lane separately and all item types together. A splitter now has only two bits of state, one for each lane, and all items of any type entering on the same lane will alternate output belts. We will always have items in the left lane.%
  \footnote{To ensure this, we make fan-ins using sideloading of the right handedness, or just sideload onto the left lane immediately before entering each gadget.}
  Also in 0.16, splitters were given a setting to sort items: items of a specified type take one exit belt, and all others take the other exit belt.

  Now our toggle switch/toggle switch for 0.15 and earlier is two independent toggle switches, and thus no longer suffices for PSPACE-hardness. Instead, we can use that item types alternate together and item sorting to construct a toggle switch/toggle switch, shown in Figure~\ref{fig:belt 0.16 tw/tw}. The agent will now be simulated by a pair of items of different types; we use an advanced circuit (``red circuit'') and a processing unit (``blue circuit''). The red circuit takes a natural path through the gadget, while the blue circuit shadows it to keep two splitters in the same state.

  \begin{figure}
    \centering
    \includegraphics[width=.7\linewidth]{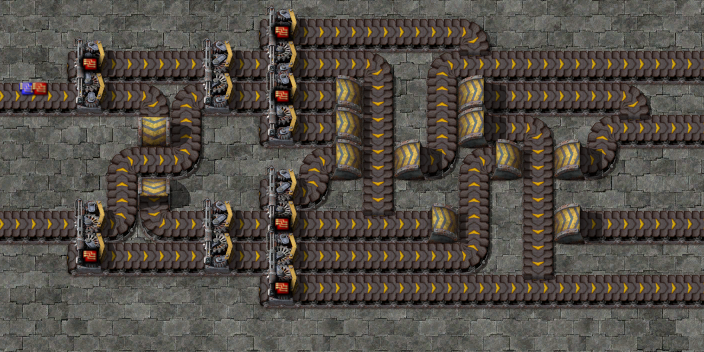}
    \caption{A toggle switch/toggle switch for transport belts in Factorio 0.16 and later. Both middle splitters encode the state of the gadget. Suppose the red and blue circuits enter the top entrance when the gadget is in the up state. The red circuit goes to the upper of the two middle splitters, takes the top exit belt, get sorted onto the topmost belt, and finally takes the topmost exit. The blue circuit visits the lower middle splitter, takes the top exit, gets sorted onto the fourth belt from the bottom (just after the splitters), and also takes the topmost exit. So both items took the topmost exit, and both middle splitters flipped state. The other cases behave similarly. This construction is due to Twan van Laarhoven.}
    \label{fig:belt 0.16 tw/tw}
  \end{figure}

  The two items take different amount of times to get through the gadget and may become separated. To fix this, after each toggle switch/toggle switch we place a \emph{grouper}, shown in Figure~\ref{fig:belt 0.16 grouper}, which reduces the distance between the items by having item which arrives first take a longer path. The amount of separation from a single traversal of a gadget is bounded, so we can keep the items a bounded distance apart using an appropriate sequence of groupers after each gadget.

  \begin{figure}
    \centering
    \includegraphics[width=.4\linewidth]{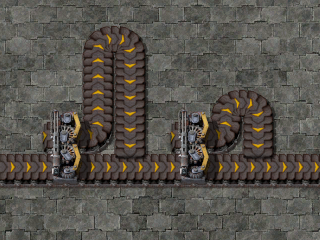}
    \caption{A grouper, which reduces the space between the red and blue circuits. The front item is delayed by about 8 tiles, and then the new front item is delayed by about 4 tiles. If the items are within about 16 tiles of each other when they enter, they exit with at most about 4 tiles between them. These distances are approximate; the actual distances are not integers since items take different amounts of time to traverse curved vs straight transport belts.}
    \label{fig:belt 0.16 grouper}
  \end{figure}

  We place a single red circuit and a single blue circuit in the left lane at the start location. Both items will reach the goal location if and only if the agent does.
\end{proof}

\begin{theorem}
  In Factorio 0.16 and later, in a world with only transport belts, underground belts, splitters, and a single item, determining whether the item reaches a specified location is in NP~$\cap$~coNP. 
\end{theorem}

\begin{proof}
  As mentioned above, a splitter with the default settings is a pair of independent toggle switches, one for each lane. A splitter set to filter will always send the item to the same output belt. Splitters have another setting also added in 0.16: they can be set to prioritize a particular input or output belt, meaning it will always use that input or output unless it is empty or blocked, respectively, instead of alternating. With a single item in the world, a splitter in priority mode always sends the item to the same output belt. Thus this problem can be reduced to zero-player motion planning with the toggle switch or equivalently ARRIVAL, which is in NP $\cap$ coNP.
\end{proof}

\section{Open Problems}
\label{sec:open problems}

One interesting problem left open by our paper and several before it \cite{Arrival,arrivalcls,switchinggames} is the complexity of zero-player motion planning with deterministic single-input input/output gadgets, or equivalently ARRIVAL and zero-player reachability switching games; this is known to be between NL-hard and NP~$\cap$~coNP, which is a large gap.
For the set switch, we do not even know NL-hardness.
We conjecture that many of these single input gadgets are P-hard and we would be interested to see such a result. We also leave open the complexity of two-player one-agent motion planning, or two-player reachability switching games, which is between PSPACE-hard and EXPTIME.

Since input/output gadgets seem to be a natural and rich class of gadgets, one could expand our characterization of zero-player motion planning to include input/output gadgets beyond those that are output-disjoint deterministic 2-state. Is there a natural notion of ``unbounded'' that implies PSPACE-hardness for a much larger class of input/output gadgets? Does every such gadget simulate the switch/set-up line/set-down line, and thus all input/output gadgets? Extending our characterization by removing any of the adjectives would be significant progress toward characterizing all input/output gadgets.

Another question we leave open is whether these gadgets remain hard in the planar case. Although our applications all contained simple crossovers, this may not always be the case, so having hardness on planar systems of gadgets would be useful.

Finally, although we only defined zero-player motion planning with input/output gadgets (and the Trainyard gadget), many other classes of gadgets could be explored in the zero-player model. This model begins to look a lot more like a typical circuit or computing model with the unusual constraint that only a single signal is ever propagating through the system. In particular, a reasonable zero-player motion planning problem with reversible deterministic gadgets (like those studied in \cite{Toggles_FUN2018} and \cite{DHL}) is similar to asynchronous ballistic reversible logic \cite{frank2017asynchronous} introduced to explore potential low-power computing architectures.

\section*{Acknowledgments}

We thank Jeffrey Bosboom for suggesting applying the gadget framework to
railroad switches (specifically, a switch/tripwire gadget) in 2017,
and Mikhail Rudoy for pointing us to the subsequent analysis of ARRIVAL
\cite{Arrival}.
We also thank Jeffrey Bosboom for providing simplified constructions for the set-up switch/set-down line and toggle switch/set-up line, and for general discussion on topics in and related to this paper.
We thank Twan van Laarhoven for providing the construction in
Figure~\ref{fig:belt 0.16 tw/tw}.
Some of this work was done during open problem solving in the MIT class on
Algorithmic Lower Bounds: Fun with Hardness Proofs (6.892)
taught by Erik Demaine in Spring 2019.
We thank the other participants of that class
for related discussions and providing an inspiring atmosphere.

This is a full version of a paper appearing in WALCOM 2022.

\bibliographystyle{alpha}
\bibliography{Bibliography}

\end{document}